\documentclass[11pt]{article}
\usepackage[margin=1.5in]{geometry}
\usepackage{amsgen,amsthm,amsmath,amstext,amsbsy,amsopn,amssymb,latexsym,appendix}
\usepackage[usenames]{color}
\usepackage[dvipsnames]{xcolor}
\usepackage{array}
\usepackage{multirow}
\usepackage{booktabs}
\usepackage{diagbox}
\usepackage{graphicx}
\usepackage{listings}
\usepackage{threeparttable}
\usepackage{tabularx}
\usepackage{array}
\usepackage{subfigure}
\usepackage{float}
\usepackage[authoryear]{natbib}
\usepackage{indentfirst}
\usepackage[linesnumbered,ruled]{algorithm2e}
\usepackage[hidelinks]{hyperref} 
\usepackage{caption}
\usepackage{subcaption}
\usepackage[dvipsnames]{xcolor}
\usepackage{makecell}
\usepackage{enumitem}
\usepackage{authblk}
\usepackage{pdflscape}
\usepackage{fancyhdr}
\usepackage{tikz}

\fancypagestyle{lscape}{
    \fancyhf{}

    \fancyfoot[C]{
        \tikz[remember picture,overlay] 
        \node[xshift=-1.5cm, rotate=90] at (current page.east) {\thepage};
    }
}
\usepackage{geometry}

\allowdisplaybreaks

\definecolor{green}{RGB}{0,160,0}

\theoremstyle{definition}

\newtheorem*{Theorem*}{Theorem}
\newtheorem*{Assumption*}{Assumption}

\begin{document}
\title{Distributionally Robust Transfer Learning with Structurally Missing Covariates}

\author[1,2]{Siqi~Li}
\author[3,4]{Chuan~Hong}
\author[3]{Ziye~Tian}
\author[5]{Benjamin~Sieu-Hon~Leong}
\author[6]{Koshi~Nakagawa}
\author[7]{Hideharu~Tanaka}
\author[8]{Sang~Do~Shin}
\author[9]{Khuong~Quoc~Dai}
\author[10]{Do~Ngoc~Son}
\author[11,12,13]{Marcus~Eng~Hock~Ong}
\author[1,2,3,13,14]{Nan~Liu\textsuperscript{$\ddagger$}}
\author[15,16]{and~Molei~Liu\textsuperscript{$\ddagger$}}
\author[ ]{for~the~PAROS~Clinical~Research~Network}

\affil[1]{Centre for Biomedical Data Science, Duke-NUS Medical School, Singapore}
\affil[2]{Duke-NUS AI + Medical Sciences Initiative, Duke-NUS Medical School, Singapore}
\affil[3]{Department of Biostatistics and Bioinformatics, Duke University, Durham, NC, USA}
\affil[4]{Duke Clinical Research Institute, Durham, NC, USA}
\affil[5]{Emergency Medicine Department, National University Hospital, Singapore}
\affil[6]{Department of Sport and Medical Science, Faculty of Physical Education, Kokushikan University, Tokyo, Japan}
\affil[7]{Graduate School of Emergency Medical System, Kokushikan University, Tokyo, Japan}
\affil[8]{Department of Emergency Medicine, Seoul National University College of Medicine, Seoul, Republic of Korea}
\affil[9]{Center for Emergency Medicine, Bach Mai Hospital, Hanoi, Vietnam}
\affil[10]{Center for Critical Care Medicine, Bach Mai Hospital, Hanoi, Vietnam}
\affil[11]{Health Services Research Centre, Singapore Health Services, Singapore}
\affil[12]{Department of Emergency Medicine, Singapore General Hospital, Singapore}
\affil[13]{Pre-hospital \& Emergency Research Centre, Health Services Research and Population Health, Duke-NUS Medical School, Singapore}
\affil[14]{NUS Artificial Intelligence Institute, National University of Singapore, Singapore}
\affil[15]{Department of Biostatistics, Peking University Health Science Center, Peking University, Beijing, China}
\affil[16]{Beijing International Center for Mathematical Research, Peking University, Beijing, China}
\date{}
\maketitle

\makeatletter
\let\@oldthefootnote\thefootnote
\let\thefootnote\relax
\begin{NoHyper}
\footnotetext{$^\ddagger$Co-last authors. Correspondence should be addressed to Molei Liu: moleiliu@bjmu.edu.cn.}
\end{NoHyper}
\let\thefootnote\@oldthefootnote
\makeatother

\begin{abstract}
Clinical prediction models trained in well-resourced centers often rely on certain measurements that are unavailable at external deployment sites due to differences in data collection infrastructure. For example, models for predicting neurological outcomes after out-of-hospital cardiac arrest (OHCA) trained in the US Resuscitation Outcomes Consortium use detailed prehospital care variables that are not recorded in many international registries. 
Standard approaches either discard these variables, sacrificing predictive information, or impute them under assumptions that cannot be verified without labeled outcome data at the deployment site.
We propose DRUM (\underline{D}istributionally \underline{R}obust \underline{U}nsupervised transfer learning with structurally \underline{M}issing covariates), a framework that transfers prediction models to new populations where certain covariates are entirely absent and outcome labels are unavailable. 
Rather than assuming a specific distribution for the missing variables, DRUM optimizes predictive performance against worst-case realizations of how these variables might behave across deployment populations, with a bias correction procedure based on Neyman orthogonality.
Applied to cross-national OHCA prediction, transferring models from the US registry to registries in Singapore, Japan, and Korea where prehospital variables are not recorded, DRUM reduced the expected calibration error by at least 45\% relative to the best-performing baseline across all three sites, producing well-calibrated risk predictions without requiring outcome data or prehospital measurements at deployment.
More broadly, the framework enables models developed in data-rich centers to be deployed reliably where key predictive variables are unavailable, advancing equitable clinical prediction across resource-limited populations.

\vspace{0.2cm}
\end{abstract}

\section{Introduction}
\label{sec.intro}
Clinical prediction models developed in high-resource research settings~\citep{wahl2018artificial} are increasingly deployed across diverse, multi-regional healthcare systems, yet ensuring their generalizability remains a fundamental challenge~\citep{de2022guidelines}. Transfer learning (TL)~\citep{li2025bridging}\footnote{Sometimes referred to as domain adaptation in the machine learning literature; we use the term transfer learning throughout this paper.}, which leverages knowledge from models trained in one setting to improve prediction in new populations, has emerged as a promising framework for adapting clinical models across heterogeneous healthcare environments~\citep{li2025leveraging}.

\begin{figure}[t]
    \centering
    \includegraphics[width=\linewidth]{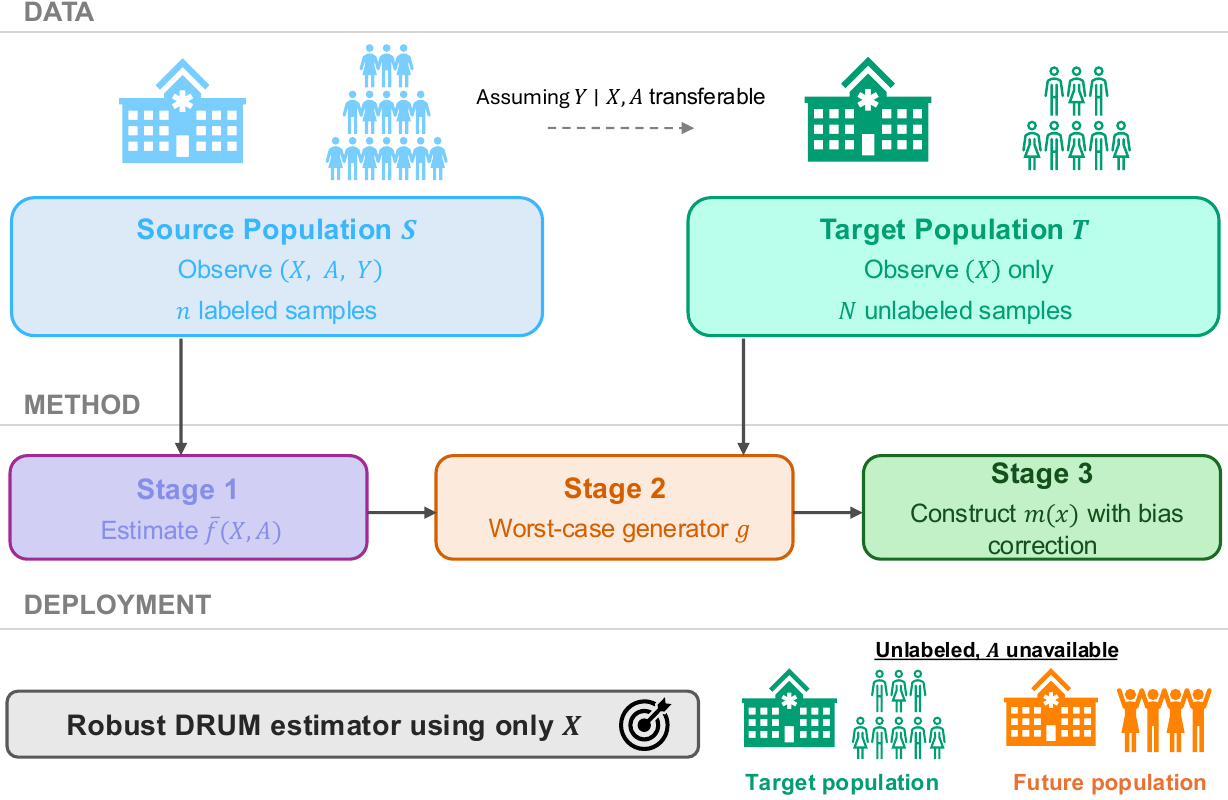}
    \caption{Overview of the DRUM framework. Source data with complete observations $(X, A, Y)$ are used to estimate the conditional mean $\bar{f}$, while unlabeled target data with only $X$ observed are used to train the worst-case generator $g$. The DRUM estimator $m(x)$ uses only covariates $X$.}
    \label{fig:framework}
\end{figure}

A common assumption underlying such knowledge transfer is that the conditional relationship between commonly observed covariates and the outcome, $Y \mid X$, remains stable across populations, so that a model trained on the source can generalize to new settings where the same covariates $X$ are observed. However, in many practical deployments, the source model is based on a richer set of covariates $(X, A)$ where certain variables $A$ that are highly predictive and are routinely recorded in the source can be completely absent in target settings due to differences in clinical infrastructure or data collection protocols. 
We refer to such variables as \emph{structurally missing covariates}: unlike standard missing data where individual observations are incomplete, structural missingness reflects a population-level absence that is itself informative, since $A$ carries outcome-relevant information not recoverable from $X$ and its absence typically coincides with a shift in $A \mid X$.

When $A$ is structurally missing in the target, prediction must rely on $\mathbb{E}[Y \mid X]$ alone, which requires integrating over the unknown distribution of $A \mid X$ in the new population. If this distribution differs from the source, as is likely when clinical practices vary across regions, directly applying the source model produces biased and miscalibrated predictions. Adding to this challenge, labeled outcome data are often unavailable, delayed, or too limited in the target population due to resource constraints or limited follow-up infrastructure, preventing retraining or supervised adaptation.

This challenge arises concretely in predicting neurological outcomes after out-of-hospital cardiac arrest (OHCA). High-performing models trained on the US Resuscitation Outcomes Consortium (ROC) registry~\citep{roc_epistry_v3} rely on detailed prehospital care measurements such as total epinephrine dose and arterial blood pH. When transferring these models to the Pan-Asian Resuscitation Outcomes Study (PAROS) network~\citep{ong2011pan}, which spans registries across diverse healthcare systems, these prehospital variables are structurally absent across all sites, and the underlying prehospital care patterns vary substantially due to different emergency response protocols and resource levels. The study populations and data structure are described in Section~\ref{sec:data}.

The OHCA setting reflects a broader methodological challenge: transfer learning when certain covariates available in the source are structurally absent in the target, and the conditional distribution $A \mid X$ in the target differs from the source in unknown ways.
Existing approaches fall short: imputation methods require assumptions about the distribution of $A$ in the target domain that cannot be verified without target labels, while covariate shift corrections only address distributional differences in observed variables and cannot recover information about covariates that are entirely absent. 
This motivates a framework that accounts for uncertainty in the unobserved $A \mid X$ without relying on untestable distributional assumptions.

\subsection{Related Work}

\paragraph{Unsupervised transfer learning.}
Following~\citet{weiss2016survey}, we operate in the \emph{unsupervised transfer learning} paradigm, characterized by abundant labeled source data but no labeled target data. Classical methods typically rely on importance reweighting \citep{shimodaira2000improving, sugiyama2007covariate}, which corrects for distributional mismatch by weighting source samples by the density ratio $dQ_X / dP_X$, with nonparametric extensions through kernel mean matching~\citep{huang2006correcting} or classifier-based estimation~\citep{bickel2009discriminative}. Representation-based approaches instead learn domain-invariant features by minimizing distributional discrepancy in a shared embedding space, for example, through maximum mean discrepancy~\citep{long2015learning} or adversarial domain classification~\citep{ganin2016domain}. 
Generative approaches such as CyCADA~\citep{hoffman2018cycada} leverage unlabeled target data through pixel- and feature-level adversarial alignment.
Despite their methodological differences, all of these methods assume a shared covariate space between the source and target. They cannot handle the case where certain outcome-relevant covariates are entirely absent in the target domain, as occurs in our setting.

\paragraph{Missing data and blockwise missingness.}
A parallel literature addresses missing covariate data through \emph{blockwise missingness}, where entire groups of variables are absent for certain data sources during multi-source integration~\citep{little2019statistical}. 
Methods in this area include imputation-based approaches~\citep{xue2021integrating}, structured-sparsity approaches that partition samples by missingness pattern~\citep{yuan2012multi, xiang2014bi}, covariance-based optimal prediction~\citep{yu2020optimal}, modular regression for multi-modal covariates~\citep{jin2023modular}, and adaptive learning under distributional shift~\citep{li2025adaptivelearningblockwisemissing}. 
Despite their differences, these methods share a common strategy which leverages sources where a variable block is observed to recover information for sources where it is absent. 
This literature diverges from our context in two critical ways. First, these methods target estimation at a primary source or across pooled datasets, rather than prediction in an entirely unlabeled target domain. Second, they share an implicit assumption that the conditional distribution $A \mid X$ remains stable across sources, so that information learned where $A$ is observed transfers directly to settings where it is absent. This assumption cannot be verified without target labels and is unlikely to hold in real-world heterogeneous data.

\paragraph{Distributionally robust optimization.}
To avoid relying on these untestable assumptions about the stability of $A \mid X$, distributionally robust optimization (DRO) offers a principled alternative by providing worst-case guarantees over a specified uncertainty set of distributions.
For an outcome $Y$ and a full set of covariates $Z$, the standard DRO formulation seeks a predictor $m(\cdot)$ that solves the minimax problem:
$ \min_{m(\cdot)} \max_{\mathbb{P} \in \mathcal{U}} \mathbb{E}_{(Z, Y) \sim \mathbb{P}} \left[ \ell(Y, m(Z)) \right]$
where $\mathcal{U}$ is an uncertainty set around the training distribution, typically constrained via divergence constraints such as $f$-divergence~\citep{hu2013kullback, duchi2021learning} or Wasserstein distance~\citep{esfahani2017datadrivendistributionallyrobustoptimization, blanchet2017quantifyingdistributionalmodelrisk}. 
This paradigm has been applied broadly, including to classification~\citep{sagawa2020distributionallyrobustneuralnetworks} and algorithmic fairness~\citep{hashimoto2018fairnessdemographicsrepeatedloss, li2025ROME}.

However, conventional DRO formulations construct uncertainty sets over the \emph{joint} distribution of all variables, treating the covariate space $Z$ as a single undifferentiated block without distinguishing between covariates that remain reliably available at deployment ($X$) and those that are structurally absent ($A$). 
In clinical deployment across heterogeneous healthcare systems, this joint-distribution approach is overly conservative and poorly targeted: the primary source of uncertainty is not the marginal distribution of the stable variables $P(X)$, but rather the unknown conditional distribution $A \mid X$ of the missing covariates.

\subsection{Contributions}
We propose DRUM (\underline{D}istributionally \underline{R}obust \underline{U}nsupervised transfer learning with structurally \underline{M}issing covariates), a framework for transferring prediction models to unlabeled target populations where certain covariates are structurally absent. 
Rather than imputing $A$ under untestable distributional assumptions, DRUM learns a prediction function that optimizes the worst-case predictive performance over an uncertainty set on the target distribution of $A \mid X$. 
Our specific contributions are as follows.

First, we develop a DRO framework that isolates uncertainty to the conditional distribution $A \mid X$ using an energy distance-based uncertainty set, governed by a single robustness parameter $\delta$. 
By parameterizing this conditional distribution through a generative model, we reduce an otherwise intractable infinite-dimensional robust optimization problem into a finite-dimensional neural network objective solvable via gradient-based training. 
This parameterization offers substantial flexibility: the framework avoids parametric assumptions on either the outcome mechanism or the conditional distribution of the missing covariates and can be combined with any sufficiently flexible and differentiable function class.

Second, to ensure reliable finite-sample performance, we derive a bias-correction procedure based on Neyman-orthogonal estimation and cross-fitting. This procedure formally reduces the estimator's sensitivity to first-order biases arising from nuisance function estimation. Extensive simulation studies show that this correction yields substantial improvements in mean and worst-case prediction errors over uncorrected DRUM and standard baselines under varying degrees of distribution shift.

Third, we demonstrate DRUM on a cross-national OHCA outcome prediction task under genuine structural missingness. Using the US-ROC registry as the labeled source data, we deploy DRUM across three distinct Asian PAROS registries (Singapore, Japan, Korea) where prehospital care variables are absent from the data collection infrastructure. 
DRUM produces consistently well-calibrated predictions across all deployment sites, reducing expected calibration error (ECE) by at least 45\% relative to the best-performing baseline at each site while maintaining competitive Brier scores.

\subsection{Organization of the Article}
The remainder of the paper is organized as follows. Section~\ref{sec:data} describes the challenge of cross-national OHCA prediction and the study populations. Section~\ref{sec:methods} formulates the DRO framework and derives the dual-form solution. Section~\ref{sec:estimation} presents the computational approach, including generator parameterization and bias correction. Section~\ref{sec:simulations} presents simulation studies. Section~\ref{sec.realdata} reports the results of applying DRUM to cross-national cardiac arrest outcome prediction. Additional experimental details and supplementary results are provided in the Appendix.

\section{Data Description}
\label{sec:data}
\subsection{Out-of-Hospital Cardiac Arrest}
\label{sec:data-ohca}

Prediction of neurological outcomes after out-of-hospital cardiac arrest (OHCA) presents precisely this challenge. OHCA prediction models are used to guide resource allocation, for example by identifying patients most likely to benefit from aggressive interventions such as mechanical circulatory support or early coronary angiography~\citep{nolan2019cardiac}, and well-calibrated predicted probabilities are essential for clinical deployment. 
Two features of the OHCA setting are relevant. 
First, the conditional relationship between clinical state at presentation and neurological outcome is plausibly stable across populations~\citep{li2025leveraging}, satisfying the Conditional Stability assumption underlying our framework (Section~\ref{sec:methods}). 
Second, the data infrastructure supporting OHCA registries differs sharply across regions: variables routinely captured in high-resource research registries, such as total epinephrine dose and arterial blood pH, are inconsistently recorded or entirely absent in many international registries, producing the structural missingness our framework addresses.

\begin{table}[!htbp]
\centering
\caption{Baseline characteristics of the OHCA study populations. Continuous variables are presented as mean (SD) or median [IQR]; categorical variables as $n$ (\%). The prehospital care variables ($A$) were available only in the US-ROC source data.}
\label{tab:baseline}
\small
\setlength{\tabcolsep}{4pt}
\begin{tabular}{llcccc}
\toprule
& & \textbf{Source} & \textbf{Target} & \multicolumn{2}{c}{\textbf{External}} \\
\cmidrule(lr){3-3} \cmidrule(lr){4-4} \cmidrule(lr){5-6}
& & US-ROC & Singapore & Japan & Korea \\
& & \textit{n\,=\,9,755} & \textit{n\,=\,5,933} & \textit{n\,=\,64,262} & \textit{n\,=\,8,484} \\
\midrule
\multicolumn{6}{l}{\textit{Stable covariates ($X$)}} \\[2pt]
Age, mean (SD) & & 63.8 (16.0) & 65.2 (18.4) & 72.4 (18.3) & 64.6 (19.1) \\[2pt]
Sex, $n$ (\%) & Female & 3,427 (35.1) & 2,072 (34.9) & 27,460 (42.7) & 3,079 (36.3) \\
& Male & 6,328 (64.9) & 3,861 (65.1) & 36,802 (57.3) & 5,405 (63.7) \\[2pt]
Rhythm, $n$ (\%) & Non-shock. & 5,519 (56.6) & 4,901 (82.6) & 59,795 (93.0) & 7,284 (85.9) \\
& Shockable & 4,236 (43.4) & 1,032 (17.4) & 4,467 (7.0) & 1,200 (14.1) \\[2pt]
Witness/CPR, $n$ (\%) & No/No & 1,869 (19.2) & 1,354 (22.8) & 21,706 (33.8) & 2,498 (29.4) \\
& No/Yes & 1,503 (15.4) & 1,057 (17.8) & 15,829 (24.6) & 1,535 (18.1) \\
& Yes/No & 2,351 (24.1) & 1,475 (24.9) & 11,530 (17.9) & 1,862 (21.9) \\
& Yes/Yes & 4,032 (41.3) & 2,047 (34.5) & 15,197 (23.6) & 2,589 (30.5) \\[4pt]
\multicolumn{6}{l}{\textit{Missing covariates ($A$, source only)}} \\[2pt]
Response time (min)  & & 5.40 [4.15, 6.90] & --- & --- & --- \\
Epinephrine dose (mg) & & 2.00 [1.00, 3.00] & --- & --- & --- \\
Blood pH & & 7.18 [7.02, 7.29] & --- & --- & --- \\[4pt]
\multicolumn{6}{l}{\textit{Outcome ($Y$)}} \\[2pt]
Neuro outcome, $n$ (\%) & Good & 1,885 (19.3) & 183 (3.1) & 2,217 (3.4) & 354 (4.2) \\
& Poor & 7,870 (80.7) & 5,750 (96.9) & 62,045 (96.6) & 8,130 (95.8) \\
\bottomrule
\end{tabular}
\end{table}

\paragraph{Source Data}
We use the Resuscitation Outcomes Consortium (ROC) Cardiac Epidemiologic Registry~\citep{roc_epistry_v3} as source data ($n = 9{,}755$). 
ROC was a North American clinical research network funded by the National Heart, Lung, and Blood Institute, designed to standardize prehospital data capture across participating EMS agencies.
As a consequence, the ROC registry provides granular prehospital care variables mandated by its standardized data dictionary.
The outcome is a favorable neurological status at hospital discharge, defined as Cerebral Performance Category 1 or 2 ($Y$, binary). Patient-level covariates ($X$) include age, sex, initial cardiac rhythm (shockable vs. non-shockable), and witness/bystander CPR status. The prehospital care variables are denoted $A$.

\paragraph{Target and External Data}
To emulate a realistic deployment scenario in which only stable covariates are available, we treat the Singapore OHCA cohort from the Pan-Asian Resuscitation Outcomes Study (PAROS) registry~\citep{ong2011pan} ($N = 5{,}933$) as unlabeled target data, using only $X$ for generator training. PAROS is a collaborative network of Asia-Pacific EMS systems whose data dictionary is leaner than ROC's: across participating sites, prehospital care variables such as drug dosing and blood gases are inconsistently captured or absent altogether~\citep{doctor2017pan}.
Although outcome labels are available in PAROS for retrospective evaluation, DRUM does not use target labels in model fitting, generator estimation, or bias correction.

External validation is performed in PAROS cohorts from Japan ($n = 64{,}262$) and Korea ($n = 8{,}484$), in which the prehospital variables $A$ are also unavailable. The three Asian cohorts also exhibit substantial heterogeneity in patient characteristics and outcomes relative to the source: the prevalence of favorable neurological outcome is $19.3\%$ in the US-ROC source population but only $3.1$--$4.2\%$ in the three PAROS cohorts (Table~\ref{tab:baseline}), reflecting differences in bystander response rates, EMS dispatch times, and post-arrest care pathways. 
This prevalence shift is consistent with the Conditional Stability Assumption, as it can arise from differences in the distribution of $(X, A)$ rather than changes in the outcome mechanism. Together, structural missingness in $A$, covariate shift in $X$, and substantial prevalence shift make the OHCA setting a stringent test of our transfer learning methodology.

\section{Methods}
\label{sec:methods}

\subsection{Problem Setup}
We consider a transfer learning problem where a prediction model will be deployed in a new target population using only covariates $X$ that are consistently observed in all settings, while certain covariates $A$ are structurally missing in the target. 
Let $\mathcal{D}_{\mathcal{S}}=\{(X_i,A_i,Y_i)\}_{i=1}^n$ denote labeled source data from population $\mathcal{S}$ and let $\mathcal{D}_{\mathcal{T}}=\{X_i\}_{i=n+1}^{n+N}$ denote unlabeled target data from population $\mathcal{T}$ where only $X$ is observed.
We make the following assumption on the stability of the outcome mechanism.

\begin{Assumption*}[Conditional Stability]
\label{ass:stability}
The conditional distribution of outcome is identical in source and target populations: $\mathbb{P}^{\mathcal{S}}_{Y \mid X, A} = \mathbb{P}^{\mathcal{T}}_{Y \mid X, A}$.
\end{Assumption*}
The marginal distribution of $X$ can shift from $\mathbb{P}_X$ (source) to $\mathbb{Q}_X$ (target), and the distribution of $A$, marginal or conditional on $X$, can also shift. Since $A$ is structurally absent in the target, conditional stability cannot be empirically verified at deployment, analogous to the covariate shift assumption in standard transfer learning~\citep{shimodaira2000improving}. Under this assumption, the only remaining source of unidentifiability is the unknown target distribution of $A \mid X$.

Our goal is to learn a predictor $m: \mathbb{R}^{d_X} \to \mathbb{R}$ that uses only $X$ and remains reliable under distributional differences between source and target. A standard empirical risk minimization (ERM) approach would estimate $\mathbb{E}[Y \mid X]$ directly from the source data, implicitly assuming that the conditional distribution $A \mid X$ in the target matches the source. When this assumption fails, the ERM predictor inherits a systematic bias that cannot be diagnosed without target labels. 
To guard against this failure mode, we adopt a distributionally robust optimization (DRO) formulation that optimizes over a set of plausible target distributions of $A \mid X$, ensuring that the resulting predictor performs well even under worst-case realizations of the missing covariates:
\begin{equation}
\max_{m(\cdot)}\min_{\mathbb{P}_{A\mid X} \in \mathcal{C}(\delta)} \mathbb{E}_{X \sim \mathbb{Q}_X, A \sim \mathbb{P}_{A\mid X}, Y \sim \mathbb{P}_{Y \mid X, A}^{\mathcal{S}}} \big[Y^2 - (Y - m(X))^2 \big],
\label{eq:objective}
\end{equation}
where $\mathcal{C}(\delta)$ is an uncertainty set of conditional distributions of $A \mid X$ controlled by a radius parameter $\delta \geq 0$, defined below. 
The objective is a squared-error-reduction criterion: $Y^2 - (Y - m(X))^2$ measures the improvement in squared prediction error achieved by $m(X)$ relative to the null predictor $m_0(x) = 0$. The zero reference is a deliberate modeling choice: it measures each predictor's utility against making no prediction, so that the worst-case objective favors predictors whose improvement over the null is robust across adverse target distributions of the missing covariates.

The uncertainty set $\mathcal{C}(\delta)$ is centered at the source conditional distribution $\mathbb{P}^{\mathcal{S}}_{A \mid X}$ and contains all conditional distributions whose deviation from the source, measured by the energy distance (defined formally in Section~\ref{sec:estimation}), does not exceed $\delta$. This construction reflects the premise that the source conditional provides a reasonable but potentially imperfect reference for the unknown target $A \mid X$. The parameter $\delta$ controls the degree of robustness: when $\delta = 0$, the uncertainty set contains only the source conditional, recovering standard transfer learning without robustness; for $\delta > 0$, it expands to include conditional distributions within a bounded neighborhood of the source, protecting against perturbations in the conditional relationship between missing and observed covariates. The choice of $\delta$ thus interpolates between trusting the source distribution and allowing progressively larger departures from it. Details of the computational implementation are given in Section~\ref{sec:estimation}.

\subsection{Dual-Form DRO Solution}
\label{sec2.2}
The minimax problem in~\eqref{eq:objective} optimizes jointly over the predictor $m(\cdot)$ and the worst-case conditional distribution $\mathbb{P}_{A\mid X}$.
Under standard regularity conditions permitting the exchange of the maximization and minimization, the following result reduces the problem to a search for a least-favorable conditional distribution of the missing covariates.
\begin{Theorem*}[Dual form of the robust predictor]
\label{thm:dual}
Under the Conditional Stability Assumption, with $\bar{f}(x,a) = \mathbb{E}[Y \mid X=x, A=a]$ bounded, the minimax problem~\eqref{eq:objective} admits the equivalent form
\begin{equation}
\min_{\mathbb{P}_{A\mid X} \in \mathcal{C}(\delta)}
\ \mathbb{E}_{X \sim \mathbb{Q}_X}
\left\{ \mathbb{E}_{A \sim \mathbb{P}_{A\mid X}}
\big[\bar{f}(X, A)\big] \right\}^2,
\label{eq:reduced}
\end{equation}
and for any conditional distribution $\mathbb{P}_{A\mid X}$ the optimal predictor is
\begin{equation}
m^*_{\mathbb{P}}(x)
= \mathbb{E}_{A \sim \mathbb{P}_{A\mid X=x}}\big[\bar{f}(x, A)\big].
\label{eq:m-optimal}
\end{equation}
Writing $\mathbb{P}^*_{A\mid X}$ for a minimizer of~\eqref{eq:reduced} whenever one exists, the robust prediction function is $m^*(x) = \mathbb{E}_{A \sim
\mathbb{P}^*_{A\mid X=x}}[\bar{f}(x, A)]$.
\end{Theorem*}

Equation~\eqref{eq:reduced} has a direct interpretation: the robust predictor is the conditional mean of $Y$ under the least favorable conditional distribution of the structurally missing covariates, and the worst case is the one that drives this predicted signal closest to the null. The derivation is given in Appendix~\ref{app:dual-derivation}.

\subsection{The Energy-Distance Uncertainty Set}
\label{sec:pop-uncertainty}

It remains to specify the uncertainty set $\mathcal{C}(\delta)$ over which the robust formulation~\eqref{eq:objective} optimizes. Because $\mathbb{E}[\bar{f}(x,A)]$ depends on the full conditional law of $A \mid X$ whenever $\bar{f}$ is nonlinear in $a$, we require a discrepancy that separates distributions rather than their means, and that is convenient to enforce as a constraint during estimation. 

The energy distance~\citep{szekely2013energy} meets both needs. 
For distributions $P$ and $Q$ on $\mathbb{R}^{d_A}$, let $A, A' \stackrel{\text{iid}}{\sim} P$ and $\tilde{A}, \tilde{A}' \stackrel{\text{iid}}{\sim} Q$ be mutually independent draws. 
The (squared) energy distance is
\begin{equation}
\mathcal{E}(P, Q)
= 2\,\mathbb{E}\|A - \tilde{A}\|_2
- \mathbb{E}\|A - A'\|_2
- \mathbb{E}\|\tilde{A} - \tilde{A}'\|_2,
\label{eq:energy-distance}
\end{equation}
where $\|\cdot\|_2$ is the Euclidean norm. 
It contrasts cross-distribution dispersion (first term) against each distribution's internal spread (subtracted terms), giving $\mathcal{E}(P, Q) \geq 0$ with equality if and only if $P = Q$. This zero-iff-equal property ensures $\mathcal{C}(\delta)$ reflects genuine distributional closeness, while each term is estimated by a sample average that is differentiable in the parameters of a generative model and hence convenient to enforce during estimation (Section~\ref{sec:estimation}).

We apply~\eqref{eq:energy-distance} conditionally, aggregating over the source covariate distribution:
\begin{equation}
D\big(\mathbb{P}_{A\mid X}, \mathbb{P}^{\mathcal{S}}_{A\mid X}\big)
= \mathbb{E}_{X \sim \mathbb{P}^{\mathcal{S}}_X}
\Big[\, \mathcal{E}\big(\mathbb{P}_{A\mid X},\,
\mathbb{P}^{\mathcal{S}}_{A\mid X}\big) \,\Big],
\qquad
\mathcal{C}(\delta)
= \Big\{ \mathbb{P}_{A\mid X} :
D\big(\mathbb{P}_{A\mid X}, \mathbb{P}^{\mathcal{S}}_{A\mid X}\big)
\leq \delta \Big\}.
\label{eq:pop-uncertainty-set}
\end{equation}
The set is centered at the source conditional $\mathbb{P}^{\mathcal{S}}_{A\mid X}$, a reference for the unidentified target conditional. The radius $\delta \geq 0$ controls robustness: at $\delta = 0$ the
set collapses to the source conditional, recovering transfer without robustness, and larger $\delta$ admits progressively larger departures of $A \mid X$ from the source.

\section{Estimation and Computation}
\label{sec:estimation}

\subsection{Generator-Based Parameterization}
\label{sec:generator}

Direct optimization of~\eqref{eq:reduced} over the infinite-dimensional space of
conditional distributions $\mathbb{P}_{A\mid X}$ is intractable, and the robust
predictor $m^*(x) = \mathbb{E}_{A \sim \mathbb{P}_{A\mid X}}[\bar{f}(x, A)]$
depends on the full conditional law of $A \mid X$, not merely its mean, whenever
$\bar{f}$ is nonlinear in $a$. We therefore parameterize $A \mid X$ through a
generative model that maps latent noise to the space of $A$, enabling Monte
Carlo approximation of $\mathbb{E}[\bar{f}(x, A)]$ for arbitrary $\bar{f}$.

Parameterizing the generator as a neural network offers considerable modeling flexibility: neural networks can approximate arbitrarily complex conditional distributions on compact domains~\citep{hornik1989multilayer}, and the resulting finite-dimensional optimization over network parameters is solvable via standard gradient-based training. 
Specifically, we introduce the generator $g_\phi: \mathbb{R}^{d_X} \times \mathbb{R}^q \to \mathbb{R}^{d_A}$ with parameters $\phi$ to model the conditional distribution $A \mid X$:
\begin{equation}
A = g_\phi(X, \epsilon), \quad \epsilon \stackrel{\text{iid}}{\sim} \mathcal{N}(0, I_q),
\label{eq:generator-local}
\end{equation}
where the latent dimension $q$ controls the expressiveness of the generator and is treated as a tuning parameter.

Substituting~\eqref{eq:generator-local} into~\eqref{eq:reduced}, the optimization becomes:
\begin{equation}
\min_{g_\phi \in \mathcal{C}_g(\delta)} \mathbb{E}_{X \sim \mathbb{Q}_X} \left[ \mathbb{E}_{\epsilon \stackrel{\text{iid}}{\sim} \mathcal{N}(0,I_q)}[\bar{f}(X, g_\phi(X, \epsilon))] \right]^2.
\label{eq:generator-obj}
\end{equation}
The robust prediction function is $m^*(x) = \mathbb{E}_\epsilon[\bar{f}(x, g^*(x, \epsilon))]$ once an optimal generator is obtained.

\subsection{Empirical Energy Constraint}
\label{sec:energy}

Section~\ref{sec:pop-uncertainty} defined the uncertainty set $\mathcal{C}(\delta)$ through the population energy distance between conditional distributions. We now give its empirical counterpart over the generator class, which turns the distributional constraint into a differentiable penalty on the network parameters $\phi$.

For a generator $g(x, \epsilon)$, the energy score $En(g)$ measures the fit of $g$ to the source conditional $\mathbb{P}^{\mathcal{S}}_{A\mid X}$:
\begin{equation}
En(g) = \mathbb{E}_{(X,A) \sim \mathbb{P}^{\mathcal{S}},\, \epsilon \stackrel{\text{iid}}{\sim}
\mathcal{N}(0,I_q)}\|A - g(X,\epsilon)\|_2 - \frac{1}{2}\mathbb{E}_{X \sim
\mathbb{P}^{\mathcal{S}}_X,\, \epsilon,\epsilon' \stackrel{\text{iid}}{\sim} \mathcal{N}(0,I_q)}
\|g(X,\epsilon) - g(X,\epsilon')\|_2,
\label{eq:energy-score}
\end{equation}
where $(X, A)$ are drawn from the source. The first term is the expected distance between generated and observed values of $A$; the second term is the expected distance between independent generated samples, which prevents the generator from collapsing to a point estimate. Minimizing~\eqref{eq:energy-score} over the generator class yields the best source fit, $g^{\mathcal{S}} = \arg\min_{g \in \mathcal{G}} En(g)$; we estimate it following~\citet{shen2024engression}.

The score and the population distance of Section~\ref{sec:pop-uncertainty} are two views of the same object. 
The excess score $En(g) - En(g^{\mathcal{S}})$ removes the term common to all generators and is, up to a constant factor, an empirical counterpart of the energy distance
$D(\mathbb{P}^{g}_{A\mid X}, \mathbb{P}^{\mathcal{S}}_{A\mid X})$ 
between the generator's conditional law and the source. 
It therefore serves as the sample analogue of the population radius, and we define the generator-based uncertainty set as
\begin{equation}
\mathcal{C}_g(\delta) = \big\{g \in \mathcal{G} : En(g) - En(g^{\mathcal{S}})
\leq \delta \big\},
\label{eq:uncertainty-set}
\end{equation}
the empirical realization of $\mathcal{C}(\delta)$ in~\eqref{eq:pop-uncertainty-set} with the proportionality constant absorbed into the choice of $\delta$.
The excess score is amenable to gradient-based optimization in the generator parameters, enabling enforcement of the constraint via primal-dual gradient methods (Section~\ref{sec:three-stage}).

\subsection{Estimation Procedure}
\label{sec:three-stage}
The dual-form solution $m^*(x) = \mathbb{E}_{A \sim \mathbb{P}^*_{A\mid X}}[\bar{f}(x, A)]$ of Section~\ref{sec2.2} decomposes into two estimands, the conditional mean $\bar{f}$ and the worst-case generator, and the estimation proceeds in three stages that recover them in turn.
\textit{Stage 1: Conditional mean estimation.}
Train a neural network $f_\psi(x, a)$ on source data $\mathcal{D}_{\mathcal{S}}$ to approximate $\bar{f}(x,a) = \mathbb{E}[Y \mid X=x, A=a]$ by minimizing the mean squared error:
\begin{equation}
\hat{\psi} = \arg\min_{\psi} \frac{1}{n} \sum_{i=1}^{n} (Y_i - f_{\psi}(X_i, A_i))^2.
\label{eq:stage1}
\end{equation}
For binary outcomes, the Stage~1 training loss is replaced by binary cross-entropy; the subsequent stages apply the squared objective to the estimated probabilities and remain unchanged.

\textit{Stage 2: Worst-case generator training.}
With $f_{\hat{\psi}}$ fixed, train the generator by minimizing the empirical version of~\eqref{eq:generator-obj}. We first learn the source conditional generator $\hat{g}^{\mathcal{S}}$ by minimizing the empirical energy
score~\eqref{eq:energy-score} on source data, then solve the energy-constrained problem via the primal-dual method:
\begin{equation}
\mathcal{L}(\phi, \lambda) = \frac{1}{N}\sum_{j=1}^N \left( \frac{1}{L}
\sum_{l=1}^L f_{\hat{\psi}}(X_j, g_\phi(X_j, \epsilon_{jl})) \right)^2
+ \lambda \big(\Delta En(g_\phi, \hat{g}^{\mathcal{S}}) - \delta\big),
\label{eq:lagrangian}
\end{equation}
where $\{\epsilon_{jl}\}_{l=1}^L \stackrel{\text{iid}}{\sim} \mathcal{N}(0, I_q)$ are Monte Carlo samples, $\{X_j\}_{j=1}^N$ are target covariates, $\lambda \geq 0$ is the dual variable, and $\Delta En(g_\phi, \hat{g}^{\mathcal{S}}) = En(g_\phi) - En(\hat{g}^{\mathcal{S}})$ is the energy gap. 
The optimization alternates between updating $\phi$ to minimize~\eqref{eq:lagrangian} and updating $\lambda$ by dual gradient ascent, $\lambda \leftarrow \max\big(0, \lambda + \eta_\lambda(\Delta En - \delta)\big)$. We write $g_{\hat{\phi}}$ for the trained generator at convergence, which Stage~3 uses to form the predictor.

\textit{Stage 3: Robust prediction.}
Given the trained generator $g_{\hat{\phi}}$, the robust predictor averages $f_{\hat{\psi}}$ over Monte Carlo draws of the generated covariates:
\begin{equation}
\hat{m}(x) = \frac{1}{L} \sum_{l=1}^L 
f_{\hat{\psi}}\big(x, g_{\hat{\phi}}(x, \epsilon_l)\big),
\qquad \{\epsilon_l\}_{l=1}^L \stackrel{\text{iid}}{\sim} \mathcal{N}(0, I_q),
\label{eq:prediction}
\end{equation}
which depends only on $X$ and can therefore be deployed in the target population where $A$ is unobserved.

\subsection{Bias Correction}
\label{sec.biascorrection}
Throughout this subsection we abbreviate the Stage-1 estimate $f_{\hat\psi}$ as $\hat{f}$ and the Stage-2 generator $g_{\hat\phi}$ as $\hat{g}$ (the preliminary generator).
The plug-in estimator $\hat{m}(x) = \frac{1}{L}\sum_l \hat{f}(x, \hat{g}(x, \epsilon_l))$ inherits the estimation error of $\hat{f}$, which is estimated in the higher-dimensional $(x, a)$ space. However the target function $m(x) = \mathbb{E}_A[\bar{f}(x, A)]$ is a function of $x$ alone and is therefore a lower-complexity object than $\bar{f}(x,a)$. The plug-in estimator does not exploit this reduction because it carries the full $(x,a)$-dimensional error of $\hat{f}$ into a prediction that ultimately depends only on $x$.

This issue arises generally when flexible models such as neural networks are used to estimate nuisance functions in a higher-dimensional space than the target parameter requires. In our setting it is compounded because $\hat{f}$ is used both to train the generator and to form the final prediction, so estimation errors accumulate systematically rather than canceling upon averaging. The worst-case optimization further amplifies this sensitivity by steering the generator toward regions where $\hat{f}$ produces extreme predictions.

Motivated by semiparametric efficiency theory~\citep{chernozhukov2016double, kennedy2023towards}, we construct a bias-corrected estimator targeting the lower-complexity object $m(x)$ rather than $\bar{f}(x,a)$. The key idea is to augment the plug-in prediction with a correction term that cancels the first-order contribution of $\bar{f}$ estimation error to the objective. Specifically, we form the pseudo-outcome
\begin{equation}
F_i = S_i \cdot \hat{\omega}(X_i, A_i)\big[Y_i - \hat{f}(X_i, A_i)\big] + r(1 - S_i) \cdot \mathbb{E}_{\epsilon}\big[\hat{f}(X_i, \hat{g}^{\mathrm{deb}}(X_i, \epsilon))\big],
\label{eq:F-pseudo}
\end{equation}
where $S_i = 1$ for source observations, $S_i = 0$ for target observations, and $r = n/N$. The first term reweights source residuals $[Y_i - \hat{f}(X_i, A_i)]$ by the density ratio $\hat{\omega}(x,a)$ between the generator-induced and source distributions of $(X, A)$, correcting for the systematic bias introduced by using $\hat{f}$ in place of $\bar{f}$; $\hat{\omega}$ is estimated by probabilistic classification. 
The second term is the plug-in prediction on target data, formed using the fold-specific generator $\hat{g}^{\mathrm{deb}}$ obtained from the debiasing procedure described in Appendix~\ref{app:debiased-details}.

All nuisance functions ($\hat{f}$, $\hat{\omega}$, $\hat{g}^{\mathrm{deb}}$) are estimated by three-fold cross-fitting to avoid overfitting bias. The construction of $\hat{g}^{\mathrm{deb}}$ and $\hat{\omega}$, and the full fold-cycling procedure, are given in Appendix~\ref{app:debiased-details} and Algorithm~\ref{alg:proposed-Debiased}.
\begin{figure}[t]
    \centering
    \includegraphics[width=\linewidth]{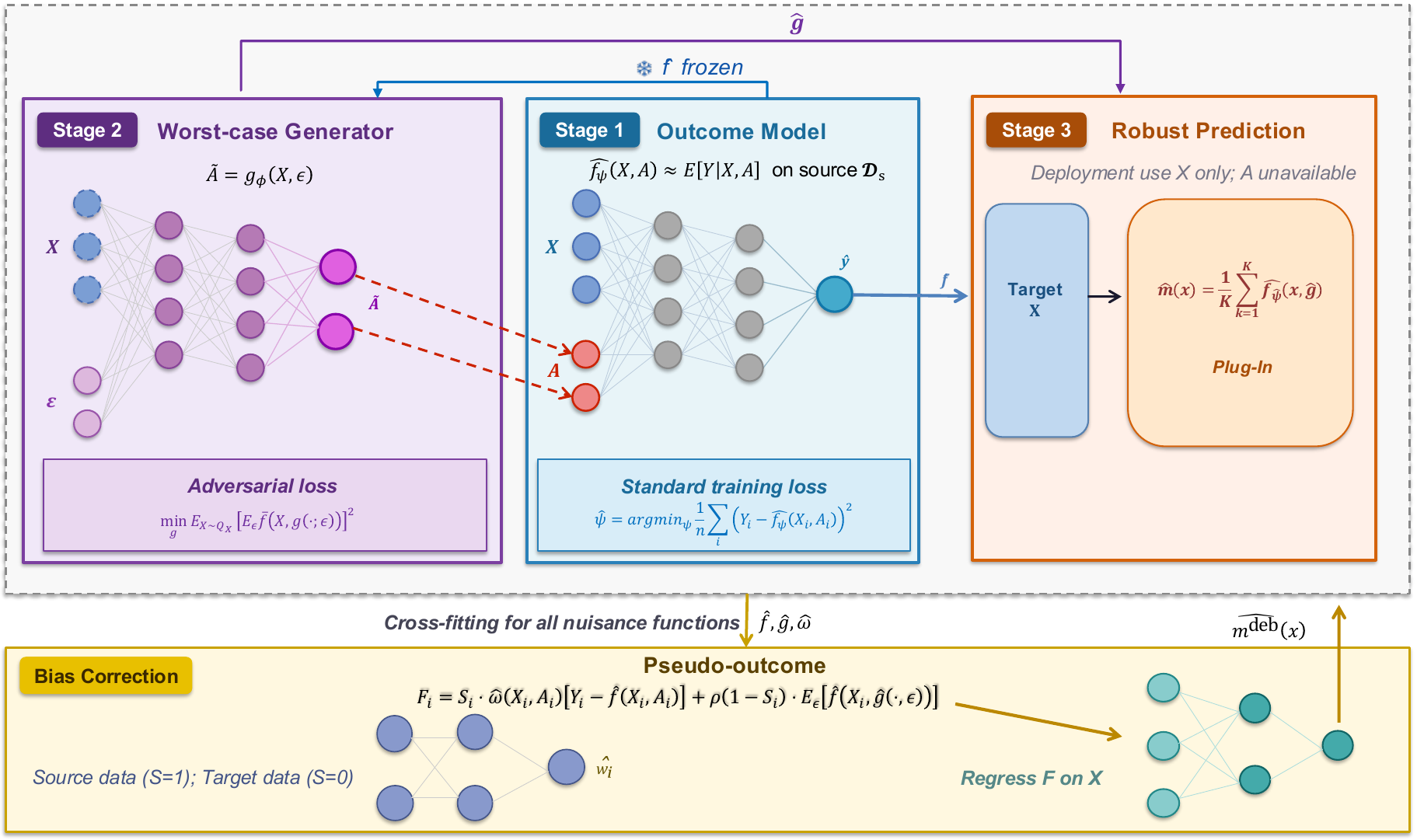}
    \caption{Estimation procedure for DRUM. Stage 1: a neural network $f_{\hat\psi}(X, A)$ is trained on labeled source data to estimate $\bar{f}(x,a)$. 
    Stage 2: with $\hat{f}$ fixed, a worst-case generator is trained on unlabeled target covariates. 
    Stage 3: the robust predictor $\hat{m}(x)$ averages $\hat{f}$ over $L$ Monte Carlo draws, producing predictions that use only $X$ at deployment. A bias correction stage regresses pseudo-outcomes $F_i$ on $X$ to reduce sensitivity to estimation error in $\hat{f}$.}
    \label{fig:drum}
\end{figure}
The bias-corrected estimator $\hat{m}^{\mathrm{deb}}(x)$ is obtained by regressing the pseudo-outcomes $F$ on $X$ using a neural network. 
The pseudo-outcome construction yields a population moment that is Neyman-orthogonal with respect to $\bar{f}$~\citep{chernozhukov2016double, kennedy2023towards} (Appendix~\ref{app:neyman}), treating the generator and density ratio as fixed, so that first-order perturbations of $\bar{f}$ do not affect the moment.
Figure~\ref{fig:drum} illustrates the full estimation procedure. Complete algorithms are provided in Algorithms~\ref{alg:local}--\ref{alg:proposed-Debiased} in Appendix~\ref{app:pseudocode}.

\section{Simulations}
\label{sec:simulations}
We evaluate DRUM through simulation studies designed to assess predictive robustness under varying degrees of conditional shift in $A \mid X$. 
We report two variants: DRUM (plug-in), the Stage-3 estimator of~\eqref{eq:prediction}, and DRUM, the bias-corrected estimator of Section~\ref{sec.biascorrection}; unless qualified, ``DRUM'' refers to the bias-corrected version.

All experiments compare DRUM against 10 baselines grouped into three categories. All methods use neural networks of comparable architecture for the final $Y \sim X$ prediction model, with hyperparameter-tuning details provided in Appendix~\ref{app:simulation-details}.

\begin{itemize}[leftmargin=*, itemsep=1pt]
    \item \textit{Baselines without target data:} Baseline-ERM, a neural network regressing $Y$ on $X$ alone via empirical risk minimization using source data only; and Baseline-DRO, the distributionally robust optimization method of \citet{duchi2021learning} applied to the same $Y \sim X$ regression.
    \item \textit{Importance weighting baselines:} IW-KMM, which estimates target/source density ratio weights $\hat{w}(x) = p_{\mathcal{T}}(x)/p_{\mathcal{S}}(x)$ via kernel mean matching \citep{huang2006correcting} using unlabeled target $X$, then trains a weighted $Y \sim X$ regression on source data; and IW-Classify, which estimates the same density ratio via a logistic regression classifier trained to distinguish source from target observations \citep{shimodaira2000improving, bickel2009discriminative}.
    \item \textit{Pseudo-label baselines:} PL-Mean, PL-MICE, and PL-MissForest, which treat the target $Y$ as missing, impute it using source $Y \mid X$ relationships (via mean imputation, MICE, or MissForest respectively), pool the imputed target data with source data, and train $Y \sim X$ models on the completed dataset. Each imputation method is combined with both ERM and DRO training, yielding six pseudo-label variants, for 10 baselines in total across the three categories.
\end{itemize}

The importance weighting and pseudo-label baselines each use unlabeled target $X$ but no target $Y$, matching the information available to DRUM. The key distinction is that all baselines model $Y \mid X$ directly, either by reweighting source observations to match the target covariate distribution or by augmenting the training set with imputed pseudo-labels. Neither approach accounts for the influence of the unobserved $A$ on $Y$: importance weighting corrects for shifts in $P(X)$ but assumes $P(Y \mid X)$ is shared across populations, while pseudo-labeling reinforces the source $Y \mid X$ relationship through the imputed values. In contrast, DRUM explicitly models the role of $A$ through $\bar{f}(X, A)$ and optimizes over the worst-case conditional distribution of $A \mid X$. The implementation details are provided in Appendix~\ref{app:simulation-details}.

\subsection{Setting I: Linear Conditional Relationship}
\label{sec:sim-setting1}

\paragraph{Data-generating Process}
Let $d_X = 15$ and $d_A = 5$. Source covariates are generated as $X_i \stackrel{\text{iid}}{\sim} \mathcal{N}(0, I_{d_X})$ for $i = 1, \ldots, n$ with $n = 5{,}000$. Covariates $A$ follow a linear conditional model:
\begin{equation}
A_i = B^\top X_i + \varepsilon_i, \quad \varepsilon_i \sim \mathcal{N}(0, \sigma_{\mathcal{S}}^2 \, I_{d_A}),
\label{eq:sim-A-source}
\end{equation}
where $B \in \mathbb{R}^{d_X \times d_A}$ is a fixed coefficient matrix (shared for all three simulation settings) with entries of decreasing magnitude: the first 10 rows contain entries ranging from $0$ to $1.0$ encoding the primary dependence of each component of $A$ on subsets of $X$, while rows 11--15 contain smaller entries introducing weaker secondary dependencies (see Appendix~\ref{app:simulation-details} for the full matrix).
The noise scale is $\sigma_{\mathcal{S}} = 0.8$.

The conditional mean function is specified as:
\begin{equation*}
\bar{f}^{\text{Setting I}}(x, a) = 0.1 \sum_{j=1}^{d_A} a_j 
+ 0.1 \sum_{j=1}^{d_A} a_j^2 
+ 0.1 \sum_{j=1}^{d_A} x_j a_j 
+ 0.3(x_1 a_2 + x_2 a_1) 
+ 0.2 \sum_{j=1}^{d_A} \mathrm{sign}(a_j)\, x_j^2,
\end{equation*}
which includes linear, quadratic, interaction, and non-smooth components in $A$; as a result, the dependence of $Y$ on $A$ cannot be captured by simple averaging over $A$. Outcomes are generated as $Y_i = \bar{f}(X_i, A_i) + \eta_i$ with $\eta_i \sim \mathcal{N}(0, 0.05^2)$.
Target covariates are drawn from a shifted distribution: $X_j \stackrel{\text{iid}}{\sim} \mathcal{N}(0.1 \cdot \mathbf{1}_{d_X},\; 1.1^2 \, I_{d_X})$ for $j = 1, \ldots, N$ with $N = 1{,}000$. Only $\{X_j\}_{j=1}^N$ is observed in the target population. 

\paragraph{Evaluations}
To evaluate robustness under varying degrees of distributional shift in $A \mid X$, we generate Monte Carlo test datasets in which the conditional relationship~\eqref{eq:sim-A-source} is perturbed using a coupled design that preserves comparability across perturbation scales.

Specifically, we draw $R=500$ direction matrices $D^{(r)}$ with entries
$D^{(r)}_{kj} \stackrel{\mathrm{iid}}{\sim} \mathrm{Uniform}(-1,1)$ and,
for each replicate, $n_{\mathrm{test}}=N=1{,}000$ noise vectors
$\varepsilon_i^{(r)} \stackrel{\mathrm{iid}}{\sim}
\mathcal{N}(0,\sigma_{\mathcal{S}}^2 I_{d_A})$.
Writing $\varepsilon^{(r)}$ for the $n_{\mathrm{test}}\times d_A$ matrix whose $i$th row is $\varepsilon_i^{(r)\top}$, the direction matrices and noise vectors are held fixed across perturbation scales.
At scale $s > 0$, the test data for replicate $r$ are generated as:
\begin{equation}
A^{(r)}_{\mathrm{test}}
=
X_{\mathrm{test}}
\big(B \odot (s \cdot D^{(r)})\big)
+
\varepsilon^{(r)},
\label{eq:sim-perturbation}
\end{equation}
where $\odot$ denotes element-wise multiplication. The parameter $s$ scales the magnitude of the perturbed conditional dependence. 
Because $\tilde{B}^{(r)} = B \odot (s \cdot D^{(r)})$ has entries of random sign, $s$ controls the magnitude of the test-time $A$--$X$ dependence rather than a signed distance from the source, so error trends across $s$ need not be monotone. 

\begin{figure}[!htbp]
    \centering
    \includegraphics[width=1\linewidth]{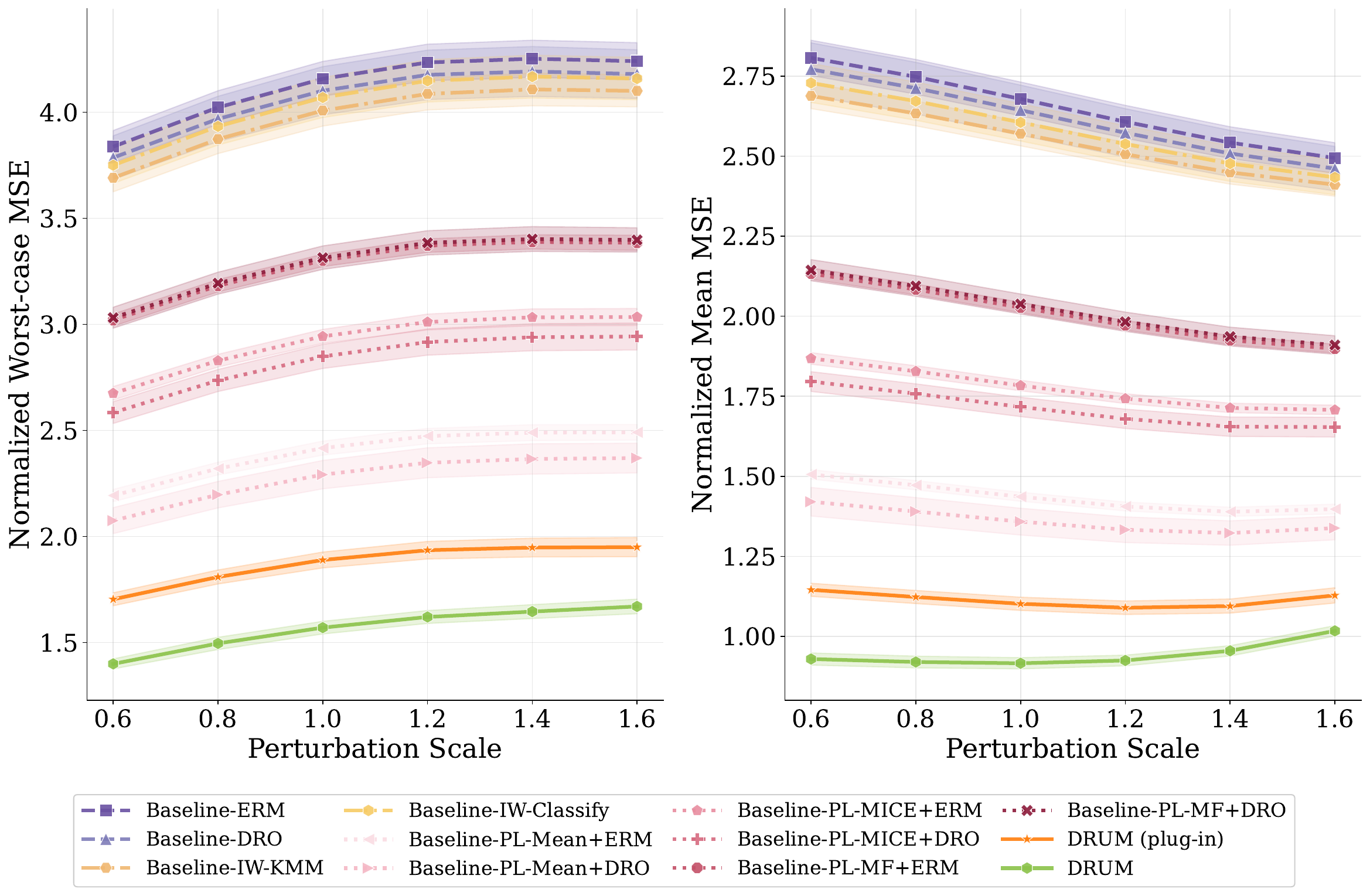}
    \caption{Worst-case and mean normalized MSE across perturbation scales $s$ for Setting~I. Lines show the mean across 10 independent replications; shaded bands indicate $\pm 1$ standard error. Each replicate is evaluated on 500 Monte Carlo test datasets per perturbation scale.
    MSE is normalized by the source outcome variance $\widehat{\mathrm{Var}}(Y^{\mathcal{S}})$.}
    \label{fig:simI}
\end{figure}

The noiseless outcomes $\bar{f}(X_{\mathrm{test}}, A^{(r)}_{\mathrm{test}})$ serve as the ground truth for evaluation.
For Setting I, we consider perturbation scales $s \in \{0.6, 0.8, 1.0, 1.2, 1.4, 1.6\}$.
The entire pipeline (source data generation, train/validation splitting, model training, and evaluation) is independently replicated across 10 random seeds.
For each method, perturbation scale, and replication, we report: 
(i)~the worst-case normalized MSE, defined as $\max_{r=1,\ldots,500}~\mathrm{MSE}_r / \widehat{\mathrm{Var}}(Y^{\mathcal{S}})$, where $\mathrm{MSE}_r$ is the mean squared error on the $r$-th Monte Carlo test set evaluated against the noiseless outcome $\bar{f}(X_{\mathrm{test}}, A^{(r)}_{\mathrm{test}})$ and $\widehat{\mathrm{Var}}(Y^{\mathcal{S}})$ is the sample variance of $Y$ in the source data; 
and (ii)~the mean normalized MSE, averaged over the 500 Monte Carlo draws. 
The worst-case metric captures robustness against adverse perturbations of $A \mid X$ in the spirit of the robust objective~\eqref{eq:objective}; the mean metric reflects average-case performance.
Results are reported as mean $\pm$ standard error across the 10 replications.

\paragraph{Results}
The results of Setting~I are presented in Figure~\ref{fig:simI}. The standard baselines (ERM, DRO) and importance reweighting baselines (IW-KMM, IW-Classify) all exhibit worst-case normalized MSE above 3.8, confirming that correcting the marginal shift in $P(X)$ does not address the conditional shift in $A \mid X$. Among the pseudo-label baselines, PL-Mean+DRO is the strongest competitor.

DRUM (plug-in) substantially outperforms all baselines on both worst-case and mean MSE, demonstrating the benefit of modeling $A \mid X$ through the generator. 
DRUM further improves over the plug-in variant, confirming that the bias correction reduces sensitivity to the estimation error of $\hat{f}$. 
DRUM also achieves the best mean MSE at every perturbation scale. Although baselines such as the pseudo-label methods embed some information about the source distribution of $A$ through imputed target outcomes, DRUM directly models $\bar{f}(x, a)$ and the conditional $A \mid X$, allowing it to adapt to perturbations in the $A \mid X$ relationship rather than treating it as fixed.

\subsection{Setting II: Nonlinear Conditional Relationship}
\label{sec:sim-setting2}
We modify Setting~I in two ways to examine performance when the conditional relationship between $A$ and $X$ is nonlinear. 
The linear mechanism $A = B^\top X + \varepsilon$ is replaced by
\begin{equation*}
A_j = \sum_{i} B_{ij} X_i + \sum_{i} 0.1 \, B_{ij} \, X_i X_{i+1} + \sum_{i} 0.1 \, B_{ij} \, \mathrm{sign}(X_i) X_i^2 + \varepsilon_j, \quad j = 1, \ldots, d_A,
\end{equation*}
where the sums run over $i = 1, \ldots, \min(5, d_X - 1)$ for interactions and $i = 1, \ldots, \min(5, d_X)$ for sign-quadratic terms, and $\varepsilon \sim \mathcal{N}(0, \sigma^2_{\mathrm{noise}} I_{d_A})$.
The outcome function adds a directional interaction:
\begin{equation}\label{eq:fbar_settingII}
\bar{f}^{\text{Setting II}}(x, a) = \bar{f}^{\text{Setting I}}(x, a) + 0.3 \sum_{j=1}^{d_A} \tanh(x_j)\, a_j,
\end{equation}
so that the effect of $A$ on $Y$ varies nonlinearly with $X$ through the $\tanh$ modulation.
We set $d_A = 2$ and $\sigma_{\mathrm{noise}} = 0.3$ (compared to $d_A = 5$ and $0.8$ in Setting~I), strengthening the signal-to-noise ratio. 
All other parameters and the Monte Carlo evaluation procedure remain identical to Setting~I, with perturbations applied to both linear and nonlinear components of $A \mid X$.

\begin{figure}[!htbp]
    \centering
    \includegraphics[width=1\linewidth]{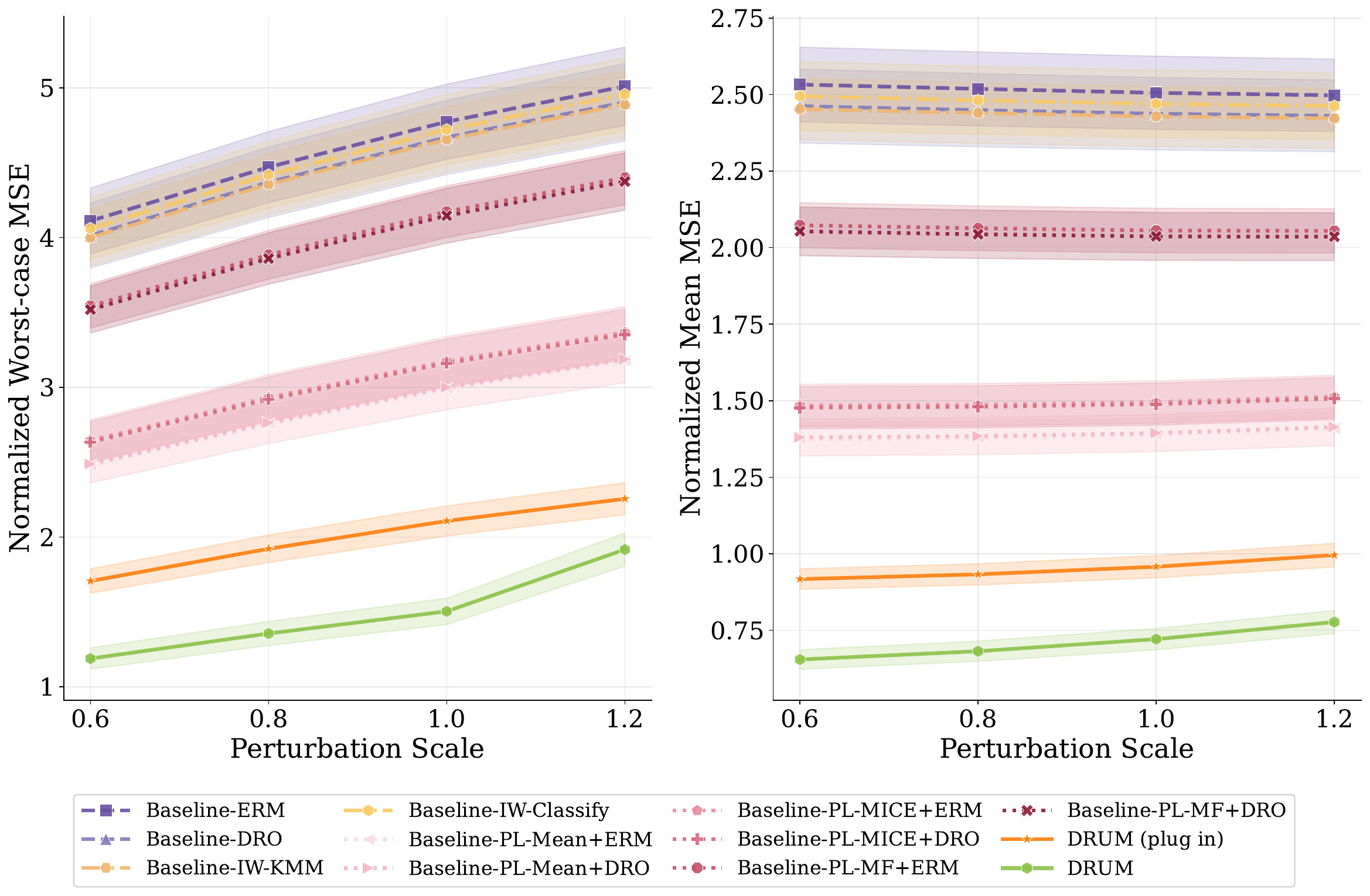}
    \caption{Worst-case and mean normalized MSE across perturbation scales $s$ for Setting~II. 
    Lines show the mean across 10 independent replications; shaded bands indicate $\pm 1$ standard error. Each replicate is evaluated on 500 Monte Carlo test datasets per perturbation scale.
    MSE is normalized by the source outcome variance $\widehat{\mathrm{Var}}(Y^{\mathcal{S}})$.}
    \label{fig:simII}
\end{figure}

\paragraph{Results}
The results for Setting~II are shown in Figure~\ref{fig:simII}. The standard baselines (ERM, DRO) and importance reweighting baselines (IW-KMM, IW-Classify) exhibit higher worst-case MSE than in Setting~I. Among the pseudo-label baselines, PL-Mean+DRO is the strongest competitor.
DRUM (plug-in) substantially outperforms all baselines on both worst-case and mean MSE. 
DRUM further improves over the plug-in variant across all perturbation scales, confirming the benefit of bias correction.
Compared to Setting~I, the worst-case MSE increases more steeply with $s$, which may reflect the more complex nonlinear $A \mid X$ relationship in this setting.

\subsection{Setting III: Varying Dimension of A}

\begin{figure}[!htbp]
    \centering
    \includegraphics[width=\linewidth]{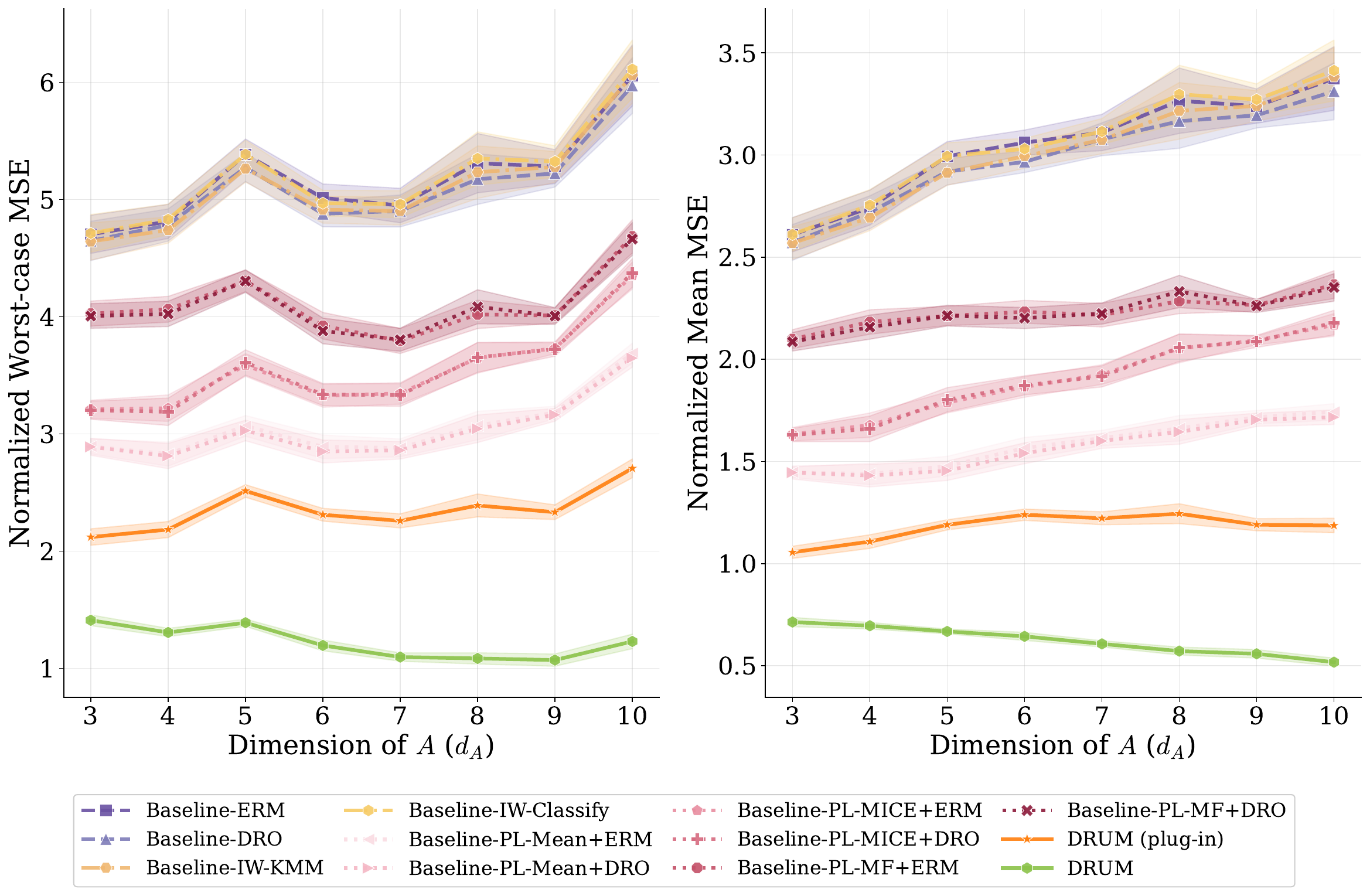}
    \caption{Worst-case and mean normalized MSE across different $d_A$ with perturbation scales fixed at $s=1.2$. 
    Lines show the mean across 10 independent replications; shaded bands indicate $\pm 1$ standard error. Each replicate is evaluated on 500 Monte Carlo test datasets for each value of $d_A$.
    MSE is normalized by the source outcome variance $\widehat{\mathrm{Var}}(Y^{\mathcal{S}})$.}
    \label{fig:simIII}
\end{figure}

We further extend Setting~II to examine how performance scales with $d_A$. We vary $d_A \in \{3, 4, 5, 6, 7, 8, 9, 10\}$ while keeping $d_X = 15$, $n = 5{,}000$, $N = 1{,}000$, $\sigma_{\mathrm{noise}} = 0.3$, and all other data-generating process parameters fixed. 
The outcome function and nonlinear $A \mid X$ mechanism are identical to Setting~II, using the first $d_A$ columns of the fixed coefficient matrix $\bar{B}$ specified in Appendix~\ref{append.simI}.
We fix the perturbation scales at $s = 1.2$ and evaluate each method on 500 Monte Carlo test datasets at each $d_A$.

\paragraph{Results}
The results of Setting~III are shown in Figure~\ref{fig:simIII}.
The standard and importance weighting baselines degrade as $d_A$ grows, with worst-case normalized MSE exceeding $5.0$ at multiple $d_A$ values.
Among the pseudo-label baselines, PL-Mean+DRO is the strongest competitor and achieves competitive worst-case MSE at several values of $d_A$. 
DRUM (plug-in) substantially outperforms all baselines on both worst-case and mean MSE. DRUM further improves over the plug-in variant across all values of $d_A$, confirming again the benefit of bias correction.

\section{Real Data Analysis}
\label{sec.realdata}

We apply DRUM to the cross-national OHCA prediction task described in Section~\ref{sec:data}, using US-ROC as labeled source data and Singapore covariates $X$ as unlabeled target data for generator training, and compare it with the same 10 baselines used in the simulations. No Singapore outcome labels enter model fitting, generator training, or bias correction; a small held-out labeled subset is used only for final hyperparameter selection, while the remaining Singapore labels and the entire Japan and Korea cohorts are reserved for evaluation. We report the bias-corrected estimator as DRUM, with full implementation details and results for the plug-in variant provided in Appendix~\ref{app:realdat-details}.

\begin{figure}[!htbp]
    \centering
    \includegraphics[width=\linewidth]{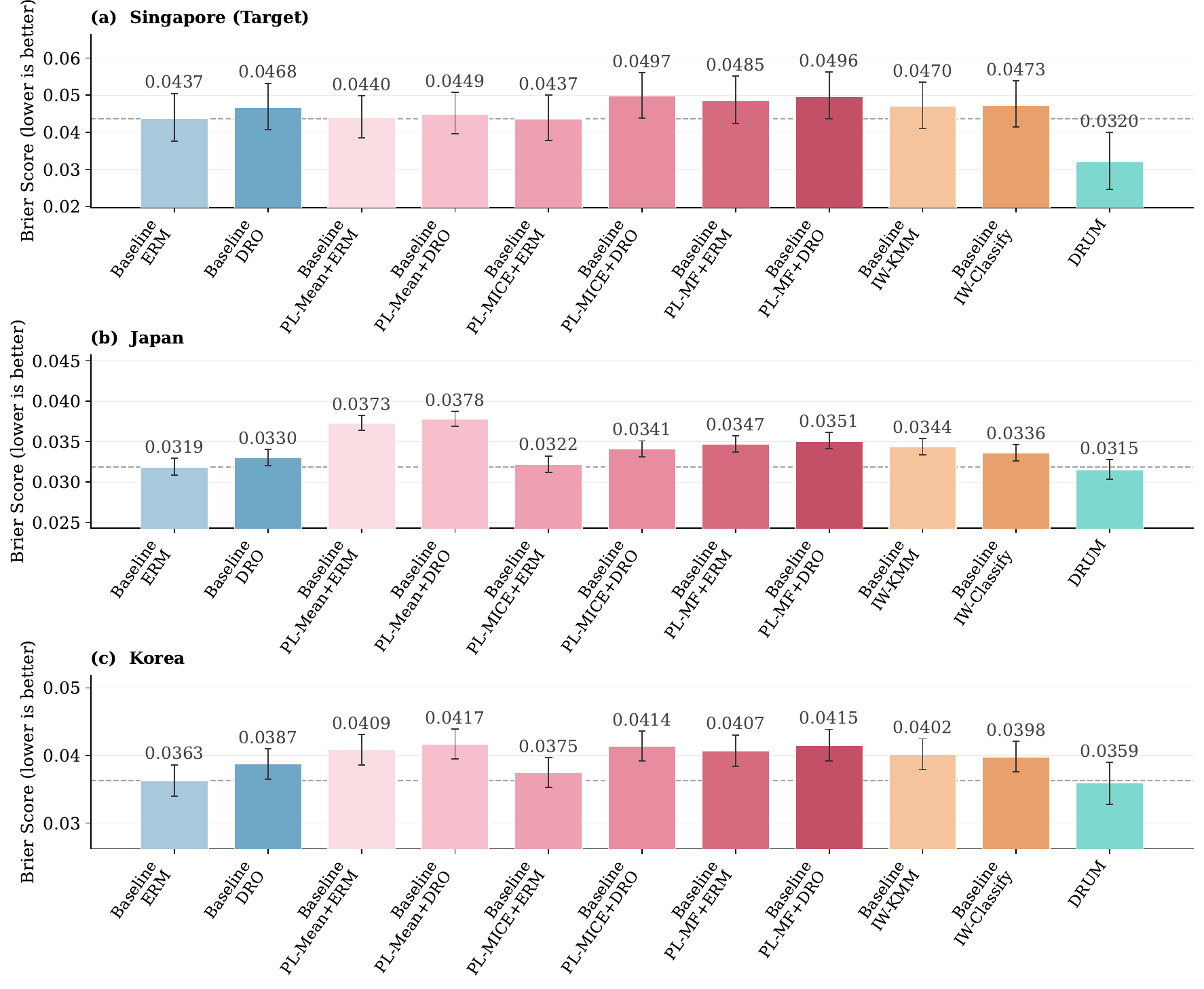}
    \caption{Brier scores across three OHCA populations. Bars show point estimates with 95\% bootstrap confidence intervals ($B=2{,}000$) for all 10 baselines and DRUM. Panel (a): target population (Singapore); panels (b)--(c): external validation cohorts (Japan and Korea). The dashed horizontal line in each panel indicates the lowest Brier score among the 10 baselines at that site.}
    \label{fig:brier_ohca_main}
\end{figure}

Figure~\ref{fig:brier_ohca_main} compares Brier scores for DRUM and all 10 baselines across the three Asian populations. DRUM attains the lowest Brier score at each site. The improvement is most pronounced in the Singapore target population, where DRUM achieves a Brier score of 0.0320 compared with 0.0437 for the best-performing baseline. In the external Japan and Korea cohorts, DRUM remains lowest but is nearly level with the strongest baseline, with Brier scores of 0.0315 versus 0.0319 in Japan and 0.0359 versus 0.0363 in Korea.

The standard DRO and importance-weighting approaches do not improve substantially over source-only ERM. Neither explicitly accounts for shifts involving the structurally missing prehospital variables $A$: importance weighting adjusts the marginal covariate shift in $P(X)$, whereas standard DRO robustifies the observed $Y \sim X$ prediction problem.
This interpretation is consistent with a source-data diagnostic reported in Appendix~\ref{app:realdat-details}: the prehospital variables are weakly predictable from $X$ but contribute approximately 40\% of the total feature importance for outcome prediction, indicating that they contain substantial information that cannot readily be recovered from the shared covariates alone.

\begin{figure}[!htbp]
    \centering
    \includegraphics[width=\linewidth]{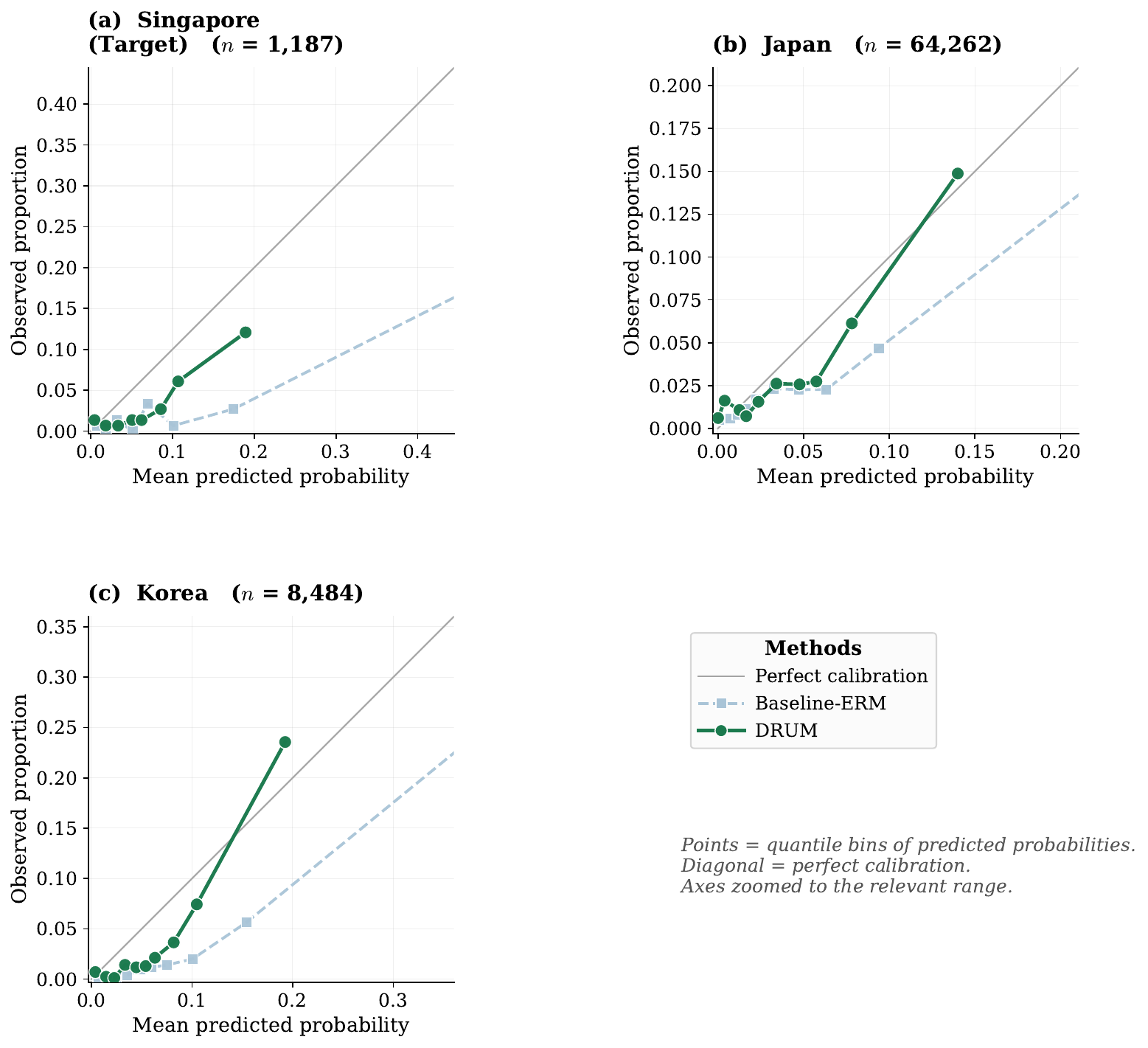}
    \caption{Calibration plots for DRUM and Baseline-ERM across three OHCA populations. Points show the mean predicted probability and observed outcome proportion within quantile bins, using 8 bins for Singapore and 10 bins for Japan and Korea. The diagonal line represents perfect calibration; axes are truncated to the relevant ranges for visibility. Baseline-ERM is shown because it achieved the lowest ECE among the 10 baselines at all three sites. Calibration curves for all methods are provided in Appendix Figure~\ref{fig:cal_all_appendix}.}
    \label{fig:cal_ohca_main}
\end{figure}

Figure~\ref{fig:cal_ohca_main} compares the calibration of DRUM with Baseline-ERM, the baseline with the lowest ECE at all three sites. Calibration measures the agreement between predicted probabilities and observed outcome frequencies and is particularly important when clinical decisions depend on absolute risk rather than patient ranking alone. For example, a model that systematically predicts risks of $15$--$20\%$ when the observed event rate is $3$--$5\%$ may misinform clinical decisions even if it ranks patients accurately.

The main-text figure is restricted to DRUM and Baseline-ERM because overlaying all 11 calibration curves would obscure the bin-level patterns; calibration curves for all methods are provided in Appendix Figure~\ref{fig:cal_all_appendix}. Baseline-ERM systematically over-predicts favorable neurological outcome across the three Asian populations, with its calibration curve lying below the perfect-calibration diagonal over most of the predicted-risk range. This pattern is consistent with the prevalence difference between the US source population, where the favorable-outcome prevalence is $19.3\%$, and the Asian populations, where it ranges from $3.1\%$ to $4.2\%$. In contrast, DRUM's calibration curve lies substantially closer to the diagonal at all three sites.

Figure~\ref{fig:ohca_ece} quantifies these differences using expected calibration error (ECE). DRUM achieves the lowest ECE among all evaluated methods and reduces ECE by at least $45\%$ relative to the best-performing baseline at each site.
ECE was computed as
\[
\mathrm{ECE}
=
\sum_{b=1}^{K}
\frac{n_b}{n}
\left|
\bar{Y}_b-\bar{p}_b
\right|,
\]
where $n_b$ is the number of observations in quantile bin $b$, $\bar{Y}_b$ is the observed outcome proportion, and $\bar{p}_b$ is the mean predicted probability in that bin. We used 8 bins for Singapore and 10 bins for Japan and Korea, consistent with the calibration plots.

\begin{figure}
    \centering
    \includegraphics[width=\linewidth]{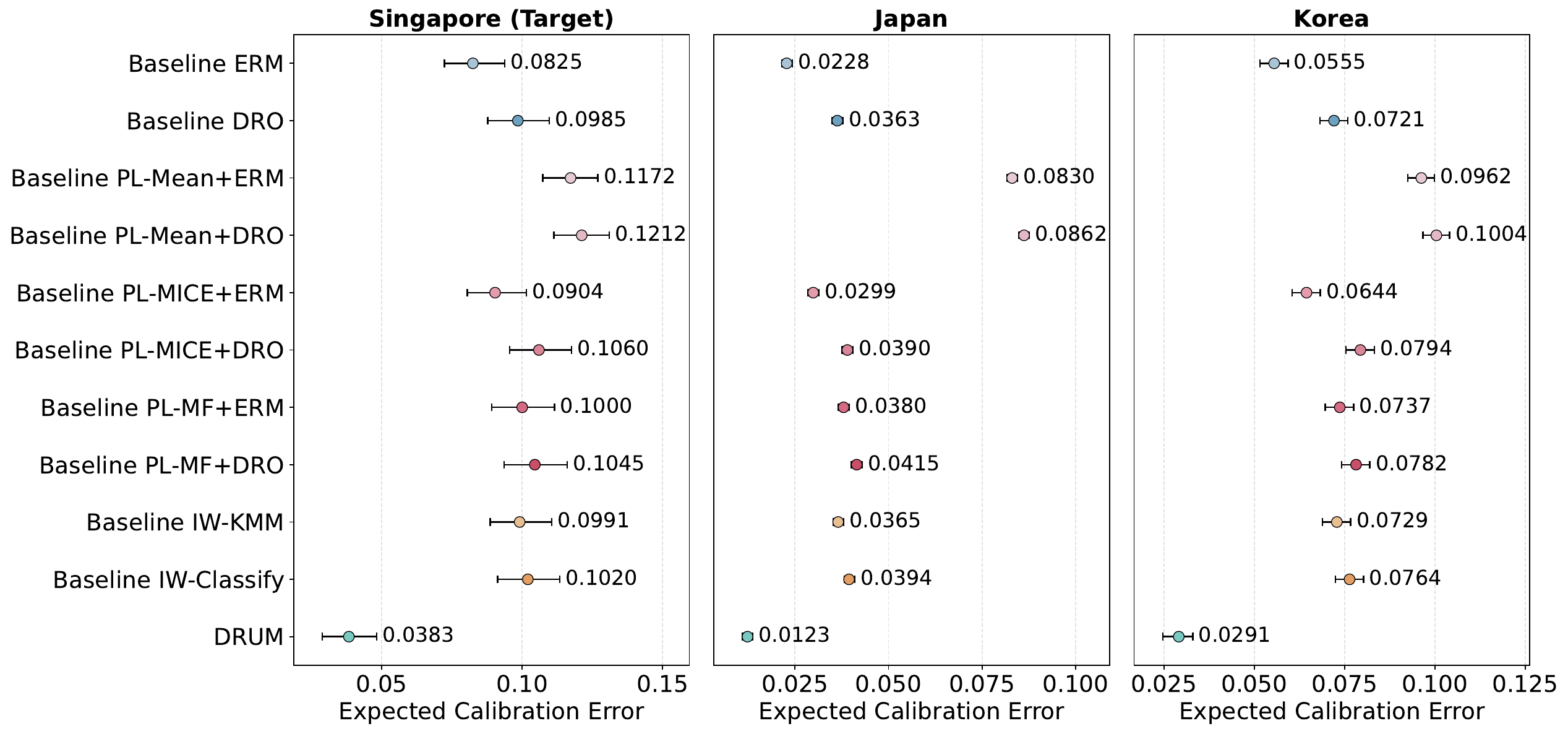}
    \caption{Expected calibration error (ECE) across three OHCA populations. Points show ECE estimates and horizontal whiskers indicate 95\% bootstrap confidence intervals ($B=2{,}000$) for all 10 baselines and DRUM; lower values indicate better calibration. ECE was computed using quantile binning with 8 bins for Singapore and 10 bins for Japan and Korea. Panel (a): target population (Singapore); panels (b)--(c): external validation cohorts (Japan and Korea).}
    \label{fig:ohca_ece}
\end{figure}

Additional discrimination metrics (AUROC, AUPRC) and fixed-cutoff classification results are reported in Appendix~\ref{app:realdat-details} (Tables~\ref{tab:discrimination_appendix}--\ref{tab:fixed-cutoff}). The plug-in estimator matches the baselines on discrimination, while the bias-corrected estimator trades some rank-based discrimination for its calibration gains. 
At clinically relevant probability cutoffs near the target prevalence (Appendix Table~\ref{tab:fixed-cutoff}), DRUM remains competitive with the best baseline on F1, precision, and specificity, indicating that the lower rank-based discrimination does not necessarily translate into materially worse classification at the thresholds relevant to clinical deployment.

Taken together, these results are consistent with the hypothesis that unobserved prehospital-care covariates contribute to cross-national miscalibration, though because $A$ is not observed in the PAROS sites the target distribution of $A \mid X$ cannot be directly identified. Regardless of the underlying mechanism, DRUM's worst-case formulation yields well-calibrated predictions across deployment sites without requiring assumptions about the target distribution of $A$.

\section{Discussion}
\label{sec:discussion}
This study addresses unsupervised transfer learning under structural covariate missingness, where outcome-relevant variables are available during source training but absent at deployment and target labels are unavailable. 
Across simulations and cross-national OHCA prediction, DRUM improved robust prediction performance by using rich source covariates to learn the outcome mechanism while producing a deployable predictor that depends only on shared covariates $X$. 

The simulation studies isolate the mechanism behind this advantage. Under controlled perturbations of the conditional distribution $A \mid X$, DRUM outperformed shift-correction and imputation baselines on both worst-case and average prediction error across all perturbation scales and dimensions of $A$, precisely the regime the baselines cannot address because they model $Y \mid X$ under a fixed conditional. The bias correction further improved over the plug-in estimator throughout, confirming that reducing sensitivity to nuisance-estimation error yields tangible finite-sample gains.

In the real-data application, DRUM's advantage is concentrated in calibration (ECE), where it improves over the best baseline at all three sites, while overall accuracy (Brier score) and threshold-based classification remain competitive with the baselines. This is consistent with the nature of the distributional shift in this application: with a low-dimensional shared covariate space, the relative ranking of patient risk is largely preserved across populations, and the primary effect of cross-national deployment is miscalibration of predicted probabilities due to the large difference in outcome prevalence between the US source (19.3\%) and Asian target populations (3--5\%). 
AUROC and AUPRC are slightly lower for the bias-corrected DRUM, which trades some rank-based discrimination for improved calibration. Detailed discrimination results are reported in Appendix~\ref{app:realdat-details}.

The structural missingness problem DRUM addresses is a recurring obstacle in clinical prediction and a systemic challenge in global health AI. Models developed in high-resource research settings~\citep{van2021unravelling} frequently rely on measurements that deployment sites lack, because clinical data infrastructure varies substantially between high- and low-resource settings~\citep{li2026enabling}. The OHCA application illustrates this concretely with prehospital care variables, but the same pattern emerges whenever prediction models cross institutional or regional boundaries with heterogeneous data collection capabilities. Existing responses are unsatisfactory: one restricts deployment to sites that mirror the source infrastructure, limiting the reach of well-developed models, while the other discards the unavailable covariates, sacrificing the predictive information they carry.

DRUM offers a principled alternative. Transfer learning methods that require identical feature sets across sites are inapplicable in this regime; DRUM instead leverages rich source data to build deployable predictors that remain robust when the deployment site lacks a subset of the training covariates, a step toward more generalizable and equitable clinical prediction across heterogeneous healthcare systems.

\paragraph{Limitations and future directions.}
Several directions for extension remain. 
First, the current generator-based parameterization assumes continuous $A$. When $A$ includes binary or categorical components, adapting the generator is nontrivial: binary variables may be handled using an expit link with a small bandwidth parameter to maintain differentiability, while categorical variables may require Gumbel-Softmax relaxations~\citep{jang2017categorical}. Extensions to high-dimensional or mixed-type $A$ remain an open direction.
Second, the framework currently assumes a single source population; when multiple source sites have heterogeneous conditional distributions, integrating DRUM into a federated learning~\citep{li2024federated} framework that leverages multiple sources while preserving data privacy remains an open direction.
Finally, the mathematical structure of DRUM is not specific to structurally missing covariates. The same formulation applies when covariates are available during training but deliberately excluded at deployment, as in algorithmic fairness settings where protected attributes must not influence predictions~\citep{li2025ROME, li2025fairfml}, or when covariates are observed but unstable across environments, as in pharmaceutical stability prediction where laboratory conditions differ from field deployment~\citep{guideline2003stability}. Exploring these connections is a promising direction.

\section*{Acknowledgments}
\paragraph{Pan-Asian Resuscitation Outcomes Study Clinical Research Network (PAROS CRN).}

Participating Site Investigators: K Kajino (Kansai Medical University, Osaka, Japan), T Nishiuchi (Graduate School of Medical Sciences and Faculty of Medicine, Kindai University, Osaka, Japan); T Tagami (Nippon Medical School Tama Nagayama Hospital, Tokyo, Japan); KJ Hong (Seoul National University Hospital, Seoul, South Korea); JH Park (Seoul National University Hospital, Seoul, South Korea); KW Lee (Keimyung University Dongsan Hospital, Daegu, South Korea); WC Cha (Samsung Medical Center, Seoul, South Korea); KJ Song (Boramae Medical Center; Seoul, South Korea); YS Ro (Seoul National University Hospital, Seoul, South Korea); JY Kim (Seoul National University Hospital, Seoul, South Korea); YJ Lee (Seoul National University Hospital, Seoul, South Korea); S Moon (Korea University Ansan Hospital, Gyeonggi, South Korea); DA Nguyen (Bach Mai Hospital, Hanoi, Viet Nam); QTA Hoang (Hue Central General Hospital, Hue, Viet Nam); TT Tra (Cho Ray Hospital, Ho Chi Minh, Vietnam);  PD Quyet (115 Emergency Center, Ho Chi Minh, Viet Nam); CQ Luong (Bach Mai Hospital, Hanoi, Viet Nam); D PFlug (Singapore Civil Defence Force, Singapore); BSH Leong (National University Hospital, Singapore); WM Ng (Ng Teng Fong General Hospital, Singapore); NE Doctor (Sengkang General Hospital, Singapore); MYC Chia (Tan Tock Seng Hospital, Singapore); HN Gan (Changi General Hospital, Singapore); L Tiah (Changi General Hospital, Singapore); WL Tay (Ng Teng Fong General Hospital, Singapore); SY Low (Sengkang General Hospital, Singapore); LP Tham (KK Women’s and Children’s Hospital, Singapore); SL Lim (National University Heart Centre Singapore, Singapore); ES Goh (Woodlands Health, Singapore); SO Cheah (Urgent Care Clinic International, Singapore); Mao DRH (Khoo Teck Puat Hospital, Singapore); YY Ng (National University Singapore, Singapore); G Nadarajan (Singapore General Hospital, Singapore); ISY Chua (Singapore General Hospital, Singapore); AFW Ho (Singapore General Hospital, Singapore); S Arulanandam (previously from Singapore Civil Defence Force, Singapore); KC Tan (previously from Singapore Civil Defence Force, Singapore); SL Chong (KK Women’s and Children’s Hospital, Singapore).

The authors would like to thank Ms Maeve Pek from Pre-hospital and Emergency Research Centre, Duke-NUS Medical School, Singapore; the late Ms Susan Yap from Department of Emergency Medicine, Singapore General Hospital, Singapore; Ms Noor Azuin, Ms Nurul Asyikin and Ms Liew Le Xuan from Unit for Pre-hospital Emergency Care, Singapore; Ms Charlene Ong and Ms Anju Devi from Accident \& Emergency, Changi General Hospital, Singapore and Ms Woo Kai Lee from Department of Cardiology, National University Heart Centre Singapore for their contributions and support to the Singapore PAROS registry; Ms Patricia Tay from Singapore Clinical Research Institute for providing secretariat support to the PAROS CRN, and Unit for Pre-hospital Emergency Care, Singapore for facilitating/implementing the initiatives that have brought about improved response to OHCA. The authors would also like to extend their sincere gratitude to all coordinators, researchers, and EMS providers from the PAROS CRN for their invaluable support of the registry.

\section*{Ethics Statement}
This study used data from two sources. The US Resuscitation Outcomes Consortium (ROC) data were approved by the National University of Singapore Institutional Review Board, which granted an exemption for this study (NUS-IRB-2023-451). The Pan-Asian Resuscitation Outcomes Study (PAROS) data were approved by the relevant ethics committees at each participating site and by the Centralized Institutional Review Board and Domain Specific Review Board for Singapore (reference numbers: 2010-270-C, C-10-545, 2013/604/C, and 2013/00929). Informed consent was waived for both datasets due to the retrospective, observational nature of the study, and all data were de-identified prior to analysis.

\section*{Data and Code Availability}
The Python implementation of the proposed algorithms is available at \url{https://github.com/siqili0325/DRUM}.
The Resuscitation Outcomes Consortium (ROC) Epistry database (Version 3, April 2011–June 2015) can be requested through the NIH website at \url{https://biolincc.nhlbi.nih.gov/studies/roc_cardiac_epistry_3/}.
The PAROS dataset contains confidential information and is governed by IRB restrictions. In accordance with these policies, the data are available to qualified researchers upon reasonable request, pending approval from the PAROS governing body.

\section*{Funding}
{\sloppy This study was supported by the Khoo Postdoctoral Fellowship Award (Duke-NUS-KPFA/2025/0081) funded by the Estate of Tan Sri Khoo Teck Puat, awarded to SL; the National Medical Research Council, Singapore, through Clinician Scientist Awards (NMRC/CSA/024/2010 and NMRC/CSA/0049/2013), and the Ministry of Health, Singapore, through the Health Services Research Grant (HSRG/0021/2012), awarded to MEHO. The funders had no role in study design, data collection, analysis, interpretation, or manuscript preparation. \par}

\section*{Conflict of Interest}
Prof Marcus EH Ong is a member of the Editorial Board of Resuscitation.
Prof Marcus EH Ong reports grants from the Laerdal Foundation, Laerdal Medical, and Ramsey Social Justice Foundation for funding of the Pan-Asian Resuscitation Outcomes Study; an advisory relationship with Global Healthcare SG, a commercial entity that manufactures cooling devices. He has a licensing agreement with ZOLL Medical Corporation and patent filed (Application no: 13/047,348) for a “Method of predicting acute cardiopulmonary events and survivability of a patient.” He is also the co-founder and scientific advisor of TIIM Healthcare, a commercial entity which develops real-time prediction and risk stratification solutions for triage. These organizations have no role in conducting this research.

\bibliographystyle{apalike}
\bibliography{library}

\newpage
\appendix

\section{Derivation of the Dual Form}
\label{app:dual-derivation}
We derive the equivalent form~\eqref{eq:reduced} and the optimal predictor~\eqref{eq:m-optimal} stated in the Dual-Form Theorem of Section~\ref{sec2.2}.
The derivation is understood to hold under standard regularity conditions ensuring that the order of maximization and minimization can be exchanged.

Throughout, we write the objective in~\eqref{eq:objective} as
\begin{equation}
\Phi(m, \mathbb{P}_{A\mid X})
= \mathbb{E}_{X \sim \mathbb{Q}_X}\,
  \mathbb{E}_{A \sim \mathbb{P}_{A\mid X}}\,
  \mathbb{E}_{Y \sim \mathbb{P}^{\mathcal{S}}_{Y\mid X, A}}
  \big[\, Y^2 - (Y - m(X))^2 \,\big],
\label{eq:app-objective}
\end{equation}
so that~\eqref{eq:objective} reads $\max_{m} \min_{\mathbb{P}_{A\mid X} \in \mathcal{C}(\delta)} \Phi(m, \mathbb{P}_{A\mid X})$.

\textbf{Step 1: Reduction of the innermost expectation over $Y$.}
Fix $x$ and $a$. Expanding the square,
\begin{equation}
Y^2 - (Y - m(x))^2 = 2\,Y\,m(x) - m(x)^2 ,
\label{eq:app-expand}
\end{equation}
which is affine in $Y$. Taking the expectation over $Y \sim \mathbb{P}^{\mathcal{S}}_{Y \mid X=x, A=a}$ and using the Conditional Stability Assumption, under which this law is shared by source and target, gives
\begin{equation}
\mathbb{E}_{Y \sim \mathbb{P}^{\mathcal{S}}_{Y\mid X=x, A=a}}
\big[\, Y^2 - (Y - m(x))^2 \,\big]
= 2\,\bar{f}(x, a)\, m(x) - m(x)^2 ,
\qquad
\bar{f}(x,a) := \mathbb{E}[Y \mid X=x, A=a].
\label{eq:app-inner-Y}
\end{equation}
The quadratic term $\mathbb{E}[Y^2 \mid x, a]$ present in each of $Y^2$ and $(Y-m)^2$ cancels exactly in the difference, so only the conditional mean $\bar{f}$ enters. Substituting~\eqref{eq:app-inner-Y} into~\eqref{eq:app-objective},
\begin{equation}
\Phi(m, \mathbb{P}_{A\mid X})
= \mathbb{E}_{X \sim \mathbb{Q}_X}\,
  \mathbb{E}_{A \sim \mathbb{P}_{A\mid X}}
  \big[\, 2\,\bar{f}(X, A)\, m(X) - m(X)^2 \,\big].
\label{eq:app-phi-fbar}
\end{equation}

\textbf{Step 2: Exchange of $\max_m$ and $\min_{\mathbb{P}}$.}
We verify the hypotheses of Sion's minimax theorem~\citep{sion1958general} for the map $(m, \mathbb{P}_{A\mid X}) \mapsto \Phi(m, \mathbb{P}_{A\mid X})$.
For fixed $\mathbb{P}_{A\mid X}$, the integrand in~\eqref{eq:app-phi-fbar} is, at each $x$, a concave (indeed strictly concave) quadratic in the scalar $m(x)$; since the objective separates across $x$ under the outer expectation, $m \mapsto \Phi(m, \mathbb{P}_{A\mid X})$ is concave on the space of predictors.
For fixed $m$, the map $\mathbb{P}_{A\mid X} \mapsto \Phi(m, \mathbb{P}_{A\mid X})$ is linear in $\mathbb{P}_{A\mid X}$, since $\Phi$ is an expectation with respect to $\mathbb{P}_{A\mid X}$ of a fixed integrand; in particular it is convex. The uncertainty set $\mathcal{C}(\delta)$ is convex, and the boundedness of $\bar{f}$ ensures $\Phi$ is finite on the relevant domain. These conditions license exchanging the order of optimization:
\begin{equation}
\max_{m(\cdot)}\ \min_{\mathbb{P}_{A\mid X} \in \mathcal{C}(\delta)}
\Phi(m, \mathbb{P}_{A\mid X})
=
\min_{\mathbb{P}_{A\mid X} \in \mathcal{C}(\delta)}\ \max_{m(\cdot)}
\Phi(m, \mathbb{P}_{A\mid X}).
\label{eq:app-swap}
\end{equation}

\textbf{Step 3: Inner maximization over $m$ in closed form.}
Fix $\mathbb{P}_{A\mid X}$ and consider the inner maximization on the right-hand side of~\eqref{eq:app-swap}. Because~\eqref{eq:app-phi-fbar} separates across $x$, the maximization can be carried out pointwise: for each $x$,
\begin{equation}
\max_{m(x) \in \mathbb{R}}\
\Big\{\, 2\, \big(\mathbb{E}_{A \sim \mathbb{P}_{A\mid X=x}}[\bar{f}(x,A)]\big)\, m(x) - m(x)^2 \,\Big\}.
\label{eq:app-pointwise}
\end{equation}
This is a strictly concave quadratic in $m(x)$; setting its derivative to zero, $2\,\mathbb{E}_{A \sim \mathbb{P}_{A\mid X=x}}[\bar{f}(x,A)] - 2\, m(x) = 0$, yields the unique maximizer
\begin{equation}
m^*_{\mathbb{P}}(x) = \mathbb{E}_{A \sim \mathbb{P}_{A\mid X=x}}[\bar{f}(x, A)],
\label{eq:app-m-optimal}
\end{equation}
which is~\eqref{eq:m-optimal}. Substituting the optimal value back, the maximized pointwise objective equals $\big(\mathbb{E}_{A \sim \mathbb{P}_{A\mid X=x}}[\bar{f}(x,A)]\big)^2$, so
\begin{equation}
\max_{m(\cdot)} \Phi(m, \mathbb{P}_{A\mid X})
= \mathbb{E}_{X \sim \mathbb{Q}_X}
  \Big(\, \mathbb{E}_{A \sim \mathbb{P}_{A\mid X}}[\bar{f}(X, A)] \,\Big)^2.
\label{eq:app-maxed}
\end{equation}

\textbf{Step 4: Reduced outer problem.}
Combining~\eqref{eq:app-swap} and~\eqref{eq:app-maxed}, the minimax problem reduces to the single minimization over the worst-case conditional distribution,
\begin{equation}
\min_{\mathbb{P}_{A\mid X} \in \mathcal{C}(\delta)}\
\mathbb{E}_{X \sim \mathbb{Q}_X}
\Big(\, \mathbb{E}_{A \sim \mathbb{P}_{A\mid X}}[\bar{f}(X, A)] \,\Big)^2,
\label{eq:app-reduced}
\end{equation}
which is~\eqref{eq:reduced}. 
Writing $\mathbb{P}^*_{A\mid X}$ for a minimizer, the robust prediction function is $m^*(x) = m^*_{\mathbb{P}^*}(x) = \mathbb{E}_{A \sim \mathbb{P}^*_{A\mid X=x}}[\bar{f}(x, A)]$ by~\eqref{eq:app-m-optimal}. \qed

\section{Debiased Generator and Density-Ratio Estimation}
\label{app:debiased-details}
This appendix details the two components of the bias correction of Section~\ref{sec.biascorrection}: the density ratio $\hat{\omega}$ used to reweight source residuals, and the debiased generator $\hat{g}^{\mathrm{deb}}$ used to form the plug-in term of the pseudo-outcome~\eqref{eq:F-pseudo}. 

As in Section~\ref{sec.biascorrection}, $\hat{f} = f_{\hat\psi}$ denotes the Stage-1 estimate and $\hat{g} = g_{\hat\phi}$ the preliminary Stage-2 generator.

\paragraph{Density-ratio estimation.}
The density ratio $\hat{\omega}(x,a)$ reweights source observations to match the distribution induced by the preliminary generator from Stage~2. Let $(X, A) \sim \mathbb{Q}^{\hat{g}}_{X,A}$ denote the joint distribution with $X \sim \mathbb{Q}_X$ and $A = \hat{g}(X, \epsilon)$. Then
\begin{equation}
\hat{\omega}(x,a) = \frac{\mathbb{Q}^{\hat{g}}_{X,A}(x,a)}{\mathbb{P}^{\mathcal{S}}_{X,A}(x,a)}
\label{eq:app-omega}
\end{equation}
is estimated via probabilistic classification. Source observations $(X_i, A_i)$ and synthetic pairs $(X_j, \hat{g}(X_j, \epsilon_j))$ are pooled, a classifier $\hat{p}(S{=}1 \mid X, A)$ is trained to distinguish the two, and the ratio is recovered as $\hat{\omega} = (1 - \hat{p})/\hat{p}$, normalized to unit mean within each cross-fitting fold.

\paragraph{Debiased generator.}
The preliminary generator $\hat{g}$ from Stage~2 is subject to the estimation error of $\hat{f}$. We therefore define a fold-specific generator $\hat{g}^{\mathrm{deb}}$ by minimizing the augmented Lagrangian
\begin{equation}
\begin{split}
\mathcal{L}(\phi, \lambda) = &\frac{1}{N}\sum_{j=1}^N \left( \frac{1}{L} \sum_{l=1}^L \hat{f}(X_j, g_\phi(X_j, \epsilon_{jl})) \right)^2\\
&+\frac{2}{n}\sum_{i=1}^n \hat{\omega}(X_i, A_i)\left( \frac{1}{L} \sum_{l=1}^L \hat{f}(X_i, \hat{g}(X_i, \epsilon_{il})) \right)\big[Y_i - \hat{f}(X_i, A_i)\big] \\
&+ \lambda \big(\Delta En(g_\phi, \hat{g}^{\mathcal{S}}) - \delta\big),
\end{split}
\label{eq:debiased-lagrangian}
\end{equation}
where $\hat{\omega}(x,a)$ is the density ratio~\eqref{eq:app-omega} estimated from the preliminary generator $\hat{g}$. 
Optimization proceeds by the same primal-dual scheme as in Stage~2 (Section~\ref{sec:three-stage}). The full fold-cycling procedure is given in Algorithm~\ref{alg:proposed-Debiased}.

\section{Neyman Orthogonality of the Pseudo-Outcome}
\label{app:neyman}

We show that the moment underlying the bias-corrected estimator is Neyman-orthogonal with respect to the outcome-regression nuisance $\bar{f}$, in the sense of \citet{chernozhukov2016double}: the Gateaux derivative of the moment with respect to $\bar{f}$, evaluated at the truth, vanishes. Throughout, the density ratio $\omega$ and the generator $g$ are held fixed at their population values, and orthogonality is established with respect to first-order perturbations of $\bar{f}$, which is the nuisance estimated in the higher-dimensional $(x,a)$ space and the primary source of plug-in bias (Section~\ref{sec.biascorrection}).

Let $S \in \{0,1\}$ be the source indicator, with $\Pr(S=1) = n/(n+N)$, and let $r = n/N$. 
Conditional on $S=1$, $(X,A) \sim \mathbb{P}^{\mathcal{S}}_{X,A}$; conditional on $S=0$, $X \sim \mathbb{Q}_X$. For the fixed generator $g$, let $\mathbb{Q}^{g}_{X,A}$ denote the induced joint law with $X \sim \mathbb{Q}_X$ and $A = g(X,\epsilon)$, $\epsilon \sim \mathcal{N}(0,I_q)$, and let
\begin{equation}
\omega(x,a) = \frac{d\mathbb{Q}^{g}_{X,A}}{d\mathbb{P}^{\mathcal{S}}_{X,A}}(x,a)
\label{eq:app-omega-def}
\end{equation}
be the density ratio, which satisfies the change-of-measure identity: for any integrable $\varphi$,
\begin{equation}
\mathbb{E}\big[\, \omega(X,A)\,\varphi(X,A) \,\big|\, S=1 \,\big]
= \mathbb{E}_{(X,A) \sim \mathbb{Q}^{g}_{X,A}}\big[\, \varphi(X,A) \,\big].
\label{eq:app-change-measure}
\end{equation}
The estimand is the robust predictor $m(x) = \mathbb{E}_{\epsilon}[\bar{f}(x, g(x,\epsilon))]$. The bias-corrected estimator regresses on $X$ the pseudo-outcome
\begin{equation}
F(\bar{f})
= S\,\omega(X,A)\big[\, Y - \bar{f}(X,A) \,\big]
  \;+\; r\,(1-S)\,\mathbb{E}_{\epsilon}\big[\, \bar{f}(X, g(X,\epsilon)) \,\big],
\label{eq:app-F}
\end{equation}
in which the dependence on the nuisance $\bar{f}$ is made explicit. We study the population moment
\begin{equation}
M(\bar{f}) = \mathbb{E}\big[\, F(\bar{f}) \,\big],
\label{eq:app-moment}
\end{equation}
the expectation targeted by the regression of $F$ on $X$.

\begin{Theorem*}[Neyman orthogonality with respect to $\bar{f}$]
\label{thm:neyman}
Suppose $\bar{f}$ is bounded and $\omega$ exists and is integrable under $\mathbb{P}^{\mathcal{S}}_{X,A}$, the standard overlap condition ensuring that the target-generated support is dominated by the source. 
Then the moment~\eqref{eq:app-moment} is Neyman-orthogonal with respect to $\bar{f}$: for every bounded direction $h(x,a)$, the Gateaux derivative of $M$ along $\bar{f} \mapsto \bar{f} + t\,h$, evaluated at the true $\bar{f}$, vanishes,
\begin{equation}
\left. \frac{\partial}{\partial t} \, M(\bar{f} + t\,h) \right|_{t=0} = 0 .
\label{eq:app-orthogonality}
\end{equation}
Consequently, first-order errors in estimating $\bar{f}$ do not contribute to the first-order bias of the population moment, with the generator $g$ and density ratio $\omega$ held fixed.
\end{Theorem*}

\begin{proof}
Fix a bounded direction $h$ and write $\bar{f}_t = \bar{f} + t\,h$. Splitting~\eqref{eq:app-moment} by the source indicator and using $\Pr(S=1) = n/(n+N)$,
\begin{equation}
M(\bar{f}_t)
= \underbrace{\Pr(S{=}1)\,\mathbb{E}\big[\, \omega(X,A)\big(Y - \bar{f}_t(X,A)\big) \,\big|\, S=1 \,\big]}_{=:\,M_{\mathcal{S}}(\bar{f}_t)}
\;+\; \underbrace{r\,\Pr(S{=}0)\,\mathbb{E}_{\epsilon}\big[\, \bar{f}_t(X, g(X,\epsilon)) \,\big]}_{=:\,M_{\mathcal{T}}(\bar{f}_t)} .
\label{eq:app-split}
\end{equation}
Both terms are affine in $t$; we differentiate and evaluate at $t=0$.

\emph{Target term.} Since $\bar{f}_t = \bar{f} + t\,h$,
\begin{equation}
\left. \frac{\partial}{\partial t}\, M_{\mathcal{T}}(\bar{f}_t) \right|_{t=0}
= r\,\Pr(S{=}0)\; \mathbb{E}_{\epsilon}\big[\, h(X, g(X,\epsilon)) \,\big]
= r\,\Pr(S{=}0)\; \mathbb{E}_{(X,A)\sim\mathbb{Q}^{g}_{X,A}}\big[\, h(X, A) \,\big].
\label{eq:app-dT}
\end{equation}

\emph{Source term.} Differentiating and then applying the change-of-measure identity~\eqref{eq:app-change-measure} with $\varphi = h$,
\begin{equation}
\left. \frac{\partial}{\partial t}\, M_{\mathcal{S}}(\bar{f}_t) \right|_{t=0}
= -\,\Pr(S{=}1)\, \mathbb{E}\big[\, \omega(X,A)\, h(X,A) \,\big|\, S=1 \,\big]
= -\,\Pr(S{=}1)\, \mathbb{E}_{(X,A)\sim\mathbb{Q}^{g}_{X,A}}\big[\, h(X, A) \,\big].
\label{eq:app-dS}
\end{equation}

\emph{Cancellation.} The two derivatives share the common factor $\mathbb{E}_{\mathbb{Q}^{g}}[h(X,A)]$, and their prefactors coincide by construction of $r = n/N$:
\begin{equation}
r\,\Pr(S{=}0) = \frac{n}{N}\cdot\frac{N}{n+N} = \frac{n}{n+N} = \Pr(S{=}1).
\label{eq:app-balance}
\end{equation}
Adding~\eqref{eq:app-dT} and~\eqref{eq:app-dS},
\begin{equation}
\left. \frac{\partial}{\partial t}\, M(\bar{f}_t) \right|_{t=0}
= \big(r\,\Pr(S{=}0) - \Pr(S{=}1)\big)\, \mathbb{E}_{(X,A)\sim\mathbb{Q}^{g}_{X,A}}\big[\, h(X, A) \,\big]
= 0 .
\label{eq:app-cancel}
\end{equation}
Since $h$ was arbitrary, $M$ is Neyman-orthogonal with respect to $\bar{f}$, establishing~\eqref{eq:app-orthogonality}.
\end{proof}

\noindent
This is the property invoked in Section~\ref{sec.biascorrection}.
Because the first-order dependence of the moment on $\bar{f}$ vanishes at the truth, with $g$ and $\omega$ held fixed, error in estimating $\bar{f}$ enters this moment only through higher-order terms.

\newpage

\section{Detailed Algorithmic Procedures}
\label{app:pseudocode}

\begin{algorithm}[H]
\caption{DRUM (plug-in)}
\label{alg:local}
\SetAlgoLined
\DontPrintSemicolon
\SetKwInOut{Input}{Input}
\SetKwInOut{Output}{Output}
\Input{Source data $\mathcal{D}_{\mathcal{S}}=\{(X_i, A_i, Y_i)\}_{i=1}^n$; target covariates $\mathcal{D}_{\mathcal{T}}=\{X_i\}_{i=n+1}^{n+N}$;
learning rates $\eta_f, \eta_{g^{\mathcal{S}}}, \eta_\phi, \eta_\lambda$; initial dual variable $\lambda_0$;
epochs $E_f, E_{g^{\mathcal{S}}}$; finetune steps $T$; 
Monte Carlo samples $L$; uncertainty radius $\delta$.}
\Output{Robust predictor $\hat{m}(\cdot)$.}
\BlankLine
\tcp{Stage 1: Estimate conditional mean $\bar{f}$}
Initialize neural network parameters $\psi$\;
\For{epoch $= 1$ to $E_f$}{
    \For{mini-batch $\{(X_i, A_i, Y_i)\}$ from $\mathcal{D}_{\mathcal{S}}$}{
        $\mathcal{L}_f(\psi) \leftarrow \frac{1}{B} \sum (Y_i - f_{\psi}(X_i, A_i))^2$\;
        $\psi \leftarrow \psi - \eta_f \nabla_{\psi} \mathcal{L}_f(\psi)$\;
    }
}
Freeze $\hat{\psi}$\;
\BlankLine
\tcp{Stage 2a: Learn source conditional generator $g^{\mathcal{S}}$}
Initialize neural network parameters $\phi$\;
\For{epoch $= 1$ to $E_{g^{\mathcal{S}}}$}{
    \For{mini-batch $\{(X_i, A_i)\}$ from $\mathcal{D}_{\mathcal{S}}$}{
        Sample $\epsilon_i, \epsilon'_i \stackrel{\text{iid}}{\sim} \mathcal{N}(0, I_q)$\;
        $\hat{A}_i \leftarrow g^{\mathcal{S}}_{\phi}(X_i, \epsilon_i)$, \quad 
        $\hat{A}'_i \leftarrow g^{\mathcal{S}}_{\phi}(X_i, \epsilon'_i)$\;
        $En(\phi) \leftarrow \frac{1}{B} \sum \|A_i - \hat{A}_i\|_2 
            - \frac{1}{2B} \sum \|\hat{A}_i - \hat{A}'_i\|_2$\;
        $\phi \leftarrow \phi - \eta_{g^{\mathcal{S}}} \nabla_{\phi} En(\phi)$\;
    }
}
Store $\phi^{\mathcal{S}} \leftarrow \phi$\;
\BlankLine
\tcp{Stage 2b: Solve energy-constrained optimization}
$\phi^* \leftarrow$ \textbf{Algorithm~\ref{alg:energy-opt}}$(f_{\hat{\psi}}, g^{\mathcal{S}}_{\phi^{\mathcal{S}}}, \mathcal{D}_{\mathcal{S}}, \mathcal{D}_{\mathcal{T}}, \delta, T)$\;
\BlankLine
\tcp{Stage 3: Prediction}
\SetKwFunction{FPredict}{Predict}
\SetKwProg{Fn}{Function}{:}{}
\Fn{\FPredict{$x$}}{
    Sample $\{\epsilon_l\}_{l=1}^L \stackrel{\text{iid}}{\sim} \mathcal{N}(0, I_q)$\;
    \Return $\hat{m}(x) = \frac{1}{L} \sum_{l=1}^L f_{\hat{\psi}}(x, g_{\phi^*}(x, \epsilon_l))$\;
}
\end{algorithm}

\begin{algorithm}[H]
\caption{Energy-Constrained Optimization (Primal-Dual)}
\label{alg:energy-opt}
\SetAlgoLined
\DontPrintSemicolon
\SetKwInOut{Input}{Input}
\SetKwInOut{Output}{Output}
\Input{Predictor $f_{\hat{\psi}}$; source generator $g^{\mathcal{S}}_{\phi^{\mathcal{S}}}$; 
source data $\mathcal{D}_{\mathcal{S}}$; 
target covariates $\mathcal{D}_{\mathcal{T}}$; uncertainty radius $\delta$; 
learning rates $\eta_\phi, \eta_\lambda$; initial dual variable $\lambda_0$; steps $T$; 
Monte Carlo samples $L$.}
\Output{Optimal generator parameters $\phi^*$.}
\BlankLine
Initialize $\phi \leftarrow \phi^{\mathcal{S}}$ \tcp*{Initialize from source generator}
Initialize $\lambda \leftarrow \lambda_0$\;
\BlankLine
\For{step $= 1$ to $T$}{
    \tcp{Primal objective on target data}
    Sample $\{\epsilon_{il}\}_{l=1}^{L} \stackrel{\text{iid}}{\sim} \mathcal{N}(0, I_q)$ for each $i \in \mathcal{D}_{\mathcal{T}}$\;
    $\mathcal{L}_{\text{obj}} \leftarrow \frac{1}{N} \sum_{i \in \mathcal{D}_{\mathcal{T}}} \left(\frac{1}{L} \sum_{l=1}^{L} f_{\hat{\psi}}(X_i, g_{\phi}(X_i, \epsilon_{il}))\right)^2$\;
    
    \BlankLine
    \tcp{Energy gap on source data}
    Sample $\epsilon_i, \epsilon'_i \stackrel{\text{iid}}{\sim} \mathcal{N}(0, I_q)$ for each $i \in \mathcal{D}_{\mathcal{S}}$\;
    $\Delta En \leftarrow En(g_\phi) - En(g^{\mathcal{S}}_{\phi^{\mathcal{S}}})$ \tcp*{Eq.~\eqref{eq:energy-score}}
    
    \BlankLine
    \tcp{Primal-dual updates}
    $\mathcal{L}(\phi, \lambda) \leftarrow \mathcal{L}_{\text{obj}} + \lambda \cdot (\Delta En - \delta)$\;
    $\phi \leftarrow \phi - \eta_\phi \nabla_{\phi} \mathcal{L}(\phi, \lambda)$\;
    $\lambda \leftarrow \max(0,\; \lambda + \eta_\lambda \cdot (\Delta En - \delta))$\;
}
\Return $\phi$\;
\end{algorithm}

\begin{algorithm}[H]
\caption{DRUM (Debiased)}
\label{alg:proposed-Debiased}
\SetAlgoLined
\DontPrintSemicolon
\SetKwInOut{Input}{Input}
\SetKwInOut{Output}{Output}
\Input{Source data $\mathcal{D}_\mathcal{S}=\{(X_i, A_i, Y_i)\}_{i=1}^n$; target data $\mathcal{D}_\mathcal{T}=\{X_i\}_{i=n+1}^{n+N}$; preliminary generator $\hat{g}$ and source generator $g^{\mathcal{S}}$ from Algorithm~\ref{alg:local}; learning rates $\eta_f, \eta_\omega, \eta_\phi, \eta_\lambda, \eta_F$; epochs $E_f, E_\omega, E_F$; finetune steps $T_g$; initial dual variable $\lambda_0$; Monte Carlo samples $L$; energy budget $\delta$.}
\Output{$\hat{m}^{\mathrm{deb}}(\cdot)$.}
\BlankLine
\tcp{Partition data into 3 folds}
Partition $\mathcal{D}_\mathcal{S}$ into folds $\mathcal{I}_1, \mathcal{I}_2, \mathcal{I}_3$\;
Partition $\mathcal{D}_\mathcal{T}$ into folds $\mathcal{J}_1, \mathcal{J}_2, \mathcal{J}_3$\;
\tcp{All fold indices below are interpreted modulo 3}
\BlankLine
\tcp{Phase 1: Cross-fitted $\hat{f}$ and residuals}
\For{$k = 1, 2, 3$ \textnormal{(indices mod 3)}}{
    Train $\hat{f}^{(k)}$ on $\mathcal{I}_k$ using MSE loss (or BCE for binary outcomes)\;
    Compute $\hat{R}_i \leftarrow Y_i - \hat{f}^{(k)}(X_i, A_i)$ for $i \in \mathcal{I}_{k+1}$\;
}
\BlankLine
\tcp{Phase 2: Cross-fitted density ratio $\hat{\omega}$}
\For{$k = 1, 2, 3$}{
    Generate $\tilde{A}_i \leftarrow \hat{g}(X_i, \epsilon_i)$, $\epsilon_i \sim \mathcal{N}(0, I_q)$ for $i \in \mathcal{J}_k$\;
    Train classifier $\hat{p}^{(k)}$ on $\{(X_i, A_i, 1)\}_{i \in \mathcal{I}_k} \cup \{(X_i, \tilde{A}_i, 0)\}_{i \in \mathcal{J}_k}$\;
    Compute $\hat{\omega}_i \leftarrow \frac{1 - \hat{p}^{(k)}(X_i, A_i)}{\hat{p}^{(k)}(X_i, A_i)}$ for $i \in \mathcal{I}_{k+1}$\;
    Normalize: $\hat{\omega}_i \leftarrow \hat{\omega}_i \,/\, \mathrm{mean}(\{\hat{\omega}_i\}_{i \in \mathcal{I}_{k+1}})$\;
}
\BlankLine
\tcp{Phase 3: Cross-fitted debiased generator}
\For{$k = 1, 2, 3$}{
    Using $\hat{f}^{(k+1)}$, the density-ratio classifier $\hat{p}^{(k+1)}$, and source data $\mathcal{I}_k$, target data $\mathcal{J}_k$:\;
    Train $\hat{g}^{(k)}_{\mathrm{deb}}$ by minimizing~\eqref{eq:debiased-lagrangian} via primal-dual updates for $T_g$ steps\;
}
\BlankLine
\tcp{Phase 4: Cross-fitted target predictions}
\For{$k = 1, 2, 3$}{
    $\hat{F}_i \leftarrow \frac{1}{L}\sum_{l=1}^{L} \hat{f}^{(k+2)}(X_i, \hat{g}^{(k+1)}_{\mathrm{deb}}(X_i, \epsilon_{il}))$ for $i \in \mathcal{J}_k$\;
}
\BlankLine
\tcp{Phase 5: Assemble pseudo-outcomes and final estimator}
$r \leftarrow n / N$\;
$F_i \leftarrow \hat{\omega}_i \cdot \hat{R}_i$ for $i \in \mathcal{D}_\mathcal{S}$ \tcp*{source}
$F_i \leftarrow r \cdot \hat{F}_i$ for $i \in \mathcal{D}_\mathcal{T}$ \tcp*{target}
Train $\hat{m}^{\mathrm{deb}}(x)$ by regressing $\{F_i\}_{i=1}^{n+N}$ on $\{X_i\}_{i=1}^{n+N}$ for $E_F$ epochs\;
\end{algorithm}

\newpage
\section{Experiment Details}
\subsection{Simulation Details}\label{app:simulation-details}

\subsubsection{Setting I.}
\label{append.simI}
\paragraph{Parameter Details}
All three simulation settings share a common base coefficient matrix $\bar{B} \in \mathbb{R}^{15 \times 10}$. Each setting uses the first $d_A$ columns: $B = \bar{B}_{:,\,1:d_A}$.

\begin{equation*}
\bar{B} = 
\begin{pmatrix}
1.0 & 0.5 & 0.3 & 0.2 & 0.1 & 0.4 & 0.2 & 0.1 & 0.3 & 0.2 \\
0.4 & 1.0 & 0.4 & 0.1 & 0.2 & 0.3 & 0.5 & 0.2 & 0.1 & 0.1 \\
0.5 & 0.3 & 1.0 & 0.5 & 0.3 & 0.2 & 0.1 & 0.4 & 0.2 & 0.3 \\
0.3 & 0.2 & 0.6 & 1.0 & 0.4 & 0.1 & 0.3 & 0.5 & 0.1 & 0.2 \\
0.2 & 0.4 & 0.2 & 0.3 & 1.0 & 0.3 & 0.2 & 0.1 & 0.5 & 0.4 \\
0.1 & 0.1 & 0.3 & 0.2 & 0.5 & 1.0 & 0.4 & 0.3 & 0.2 & 0.1 \\
0.0 & 0.2 & 0.1 & 0.4 & 0.3 & 0.3 & 1.0 & 0.2 & 0.4 & 0.2 \\
0.1 & 0.0 & 0.2 & 0.1 & 0.2 & 0.2 & 0.3 & 1.0 & 0.1 & 0.3 \\
0.0 & 0.2 & 0.0 & 0.2 & 0.1 & 0.1 & 0.2 & 0.3 & 1.0 & 0.2 \\
0.0 & 0.3 & 0.0 & 0.0 & 0.1 & 0.2 & 0.1 & 0.2 & 0.3 & 1.0 \\
0.3 & 0.2 & 0.2 & 0.0 & 0.4 & 0.1 & 0.0 & 0.1 & 0.2 & 0.3 \\
0.1 & 0.4 & 0.0 & 0.3 & 0.1 & 0.2 & 0.1 & 0.0 & 0.1 & 0.2 \\
0.0 & 0.2 & 0.3 & 0.1 & 0.2 & 0.0 & 0.2 & 0.1 & 0.0 & 0.1 \\
0.2 & 0.0 & 0.1 & 0.4 & 0.0 & 0.1 & 0.0 & 0.2 & 0.1 & 0.0 \\
0.1 & 0.3 & 0.2 & 0.0 & 0.3 & 0.0 & 0.1 & 0.0 & 0.2 & 0.1
\end{pmatrix}
\label{eq:beta-matrix}
\end{equation*}

\paragraph{Implementation Details}
All hyperparameters were selected via grid search using a held-out validation set comprising 20\% of the source data (because only source data are labeled). For each method and each hyperparameter configuration, models were trained on the remaining 80\% of source data, and the configuration minimizing validation mean squared error was selected. The tuning was performed once on the source data and the selected hyperparameters were held fixed across all perturbation scales, Monte Carlo test datasets, and training replications.

\begin{table}[!htbp]
\centering
\small
\begin{tabular}{lll}
\hline
Method & Parameter & Search range \\
\hline
\multirow{5}{*}{DRUM (plug-in)} 
  & $\delta$ & $\{0.15, 0.2, 0.3\}$ \\
  & $\eta_\phi$ & $\{10^{-5}, 10^{-4}, 2 \times 10^{-4}, 10^{-3}\}$ \\
  & $\eta_\lambda$ & $\{10^{-5}, 10^{-4}, 10^{-3}\}$ \\
  & hidden dim & $\{4, 8, 16, 32\}$ \\
  & noise dim & $\{4, 8, 16, 32\}$ \\
\hline
\multirow{4}{*}{DRUM (debiased)} 
  & $\eta_\omega$ & $\{10^{-5}, 10^{-4}, 10^{-3}\}$ \\
  & $\eta_F$ & $\{10^{-5}, 10^{-4}, 10^{-3}\}$ \\
  & $E_\omega$ & $\{100, 200, 300\}$ \\
  & $E_F$ & $\{50, 100, 200, 300\}$ \\
\hline
\multirow{2}{*}{\shortstack[l]{Baseline-ERM, Baseline-IW-KMM,\\Baseline-IW-Classify, Baseline-PL+ERM}} 
  & learning rate & $\{5 \times 10^{-5}, 10^{-4}, 5 \times 10^{-4}, 10^{-3}\}$ \\
  & epochs & $\{20, 30, 50\}$ \\
\hline
\multirow{3}{*}{\shortstack[l]{Baseline-DRO,\\Baseline-PL+DRO}} 
  & learning rate & $\{5 \times 10^{-5}, 10^{-4}, 5 \times 10^{-4}, 10^{-3}\}$ \\
  & epochs & $\{20, 30, 50\}$ \\
  & $\rho$ & $\{0.25, 0.5, 0.75, 1.0, 1.25, 1.5\}$ \\
\hline
\end{tabular}
\caption{Hyperparameter search grids for all methods in Setting~I. ERM-based methods search over 12 configurations; DRO-based methods search over 72 configurations.}
\label{tab:tuning-grids-settingI}
\end{table}

\textit{DRUM methods.} For the conditional mean estimator $f_{\hat\psi}$ (shared across DRUM variants): learning rate $\eta_f = 10^{-5}$, $E_f = 100$ epochs, two hidden layers of width 128 with ReLU activations.
For DRUM (plug-in): source Engression trained with learning rate $\eta_{g^{\mathcal{S}}} = 5 \times 10^{-4}$, hidden dimension 16, noise dimension 32, $E_{g^{\mathcal{S}}} = 500$ epochs; energy-constrained optimization with primal learning rate $\eta_\phi = 2 \times 10^{-4}$, dual learning rate $\eta_\lambda = 10^{-4}$, gradient clipping at 2.0, $\delta = 0.3$, 80 fine-tuning steps.
For DRUM (debiased): density ratio classifier learning rate $\eta_\omega = 10^{-5}$, $E_\omega = 100$ epochs; final estimator learning rate $\eta_F = 10^{-5}$, $E_F = 200$ epochs; three-fold cross-fitting for all nuisance estimates.

\textit{Standard baselines.} For Baseline-ERM: selected learning rate $10^{-3}$, 20 epochs. For Baseline-DRO~\citep{duchi2021learning}: selected learning rate $5 \times 10^{-4}$, 50 epochs, $\rho = 0.25$. Both use the same architecture as $f_{\hat\psi}$ but take only $X$ as input.

\textit{Importance weighting baselines.} For IW-KMM, density ratio weights were estimated via kernel mean matching~\citep{huang2006correcting} using a Gaussian kernel with bandwidth selected by the median heuristic, upper bound $B = 1{,}000$, and constraint tolerance $\epsilon = (\sqrt{m} - 1)/\sqrt{m}$. For IW-Classify, density ratio weights were estimated via a logistic regression classifier trained to distinguish source from target observations using $X$ only. Both methods produce per-observation weights $\hat{w}(x_i)$ that are used to train a weighted MSE regression on source data. The effective sample size (ESS $= (\sum_i w_i)^2 / \sum_i w_i^2$, measuring the effective number of independent observations after reweighting) was $1{,}576/4{,}000$ for KMM and $3{,}539/4{,}000$ for the classifier, indicating moderate weight concentration. Selected hyperparameters: IW-KMM (lr $= 10^{-3}$, 30 epochs), IW-Classify (lr $= 5 \times 10^{-4}$, 50 epochs).

\textit{Pseudo-label baselines.} Target outcomes $Y$ were treated as missing and imputed using source $Y \mid X$ relationships via three methods: mean imputation, MICE (iterative imputation with Bayesian ridge regression, 20 iterations), and MissForest (iterative imputation with random forests, 100 trees, max depth 10, 25 iterations). The imputed target observations were pooled with source data and used to train $Y \sim X$ models via ERM or DRO, yielding six variants. The imputation hyperparameters were set to standard defaults; tuning these would further favor the pseudo-label baselines.

All models were trained with the Adam optimizer and batch size 128. Prediction samples at test time: $L = 256$.
Table~\ref{tab:tuning-grids-settingI} summarizes the hyperparameter search ranges for each method.

\subsubsection{Setting II}
\label{append.simII}

DRUM methods use the same hyperparameter search procedure as Setting~I. For a simplified fair comparison, the two original baselines use fixed hyperparameters: Baseline-ERM with learning rate $10^{-3}$, 50 epochs; Baseline-DRO with learning rate $10^{-3}$, 50 epochs, $\rho = 0.5$. Both use the same architecture as $f_{\hat\psi}$ but take only $X$ as input.

For the additional baselines, the same tuning grids as Setting~I (Table~\ref{tab:tuning-grids-settingI}) were used, with hyperparameters selected by validation MSE on held-out source data. Selected hyperparameters: IW-KMM (lr $= 10^{-3}$, 50 epochs), IW-Classify (lr $= 10^{-3}$, 50 epochs), PL-Mean+ERM (lr $= 10^{-3}$, 20 epochs), PL-Mean+DRO (lr $= 10^{-3}$, 20 epochs, $\rho = 0.25$), PL-MICE+ERM (lr $= 10^{-3}$, 30 epochs), PL-MICE+DRO (lr $= 10^{-3}$, 30 epochs, $\rho = 0.25$), PL-MF+ERM (lr $= 10^{-3}$, 50 epochs), PL-MF+DRO (lr $= 10^{-3}$, 50 epochs, $\rho = 0.25$). All DRO pseudo-label variants selected the minimum $\rho = 0.25$, indicating negligible benefit from the DRO penalty on pooled imputed data.

\subsubsection{Setting III}
\label{appendix:settingIII}

\begin{table}[!htbp]
\centering
\small
\begin{tabular}{lll}
\hline
Method & Parameter & Search range \\
\hline
\multirow{3}{*}{Baseline-ERM} 
  & architecture & $[128,128],\; [256,128],\; [128,128,64],\; [256,256,128]^*$ \\
  & learning rate & $\{5 \times 10^{-5}, 10^{-4}, 5 \times 10^{-4}, 10^{-3}\}$ \\
  & epochs & $\{20, 30, 50\}$ \\
\hline
\multirow{4}{*}{Baseline-DRO} 
  & architecture & $[128,128],\; [256,128],\; [128,128,64],\; [256,256,128]^*$ \\
  & learning rate & $\{5 \times 10^{-5}, 10^{-4}, 5 \times 10^{-4}, 10^{-3}\}$ \\
  & epochs & $\{20, 30, 50\}$ \\
  & $\rho$ & $\{0.25, 0.5, 0.75, 1.0, 1.25, 1.5\}$ \\
\hline
\multirow{2}{*}{All ERM baselines$^\dagger$} 
  & learning rate & $\{5 \times 10^{-5}, 10^{-4}, 5 \times 10^{-4}, 10^{-3}\}$ \\
  & epochs & $\{20, 30, 50\}$ \\
\hline
\multirow{3}{*}{All DRO baselines$^\dagger$} 
  & learning rate & $\{5 \times 10^{-5}, 10^{-4}, 5 \times 10^{-4}, 10^{-3}\}$ \\
  & epochs & $\{20, 30, 50\}$ \\
  & $\rho$ & $\{0.25, 0.5, 0.75, 1.0, 1.25, 1.5\}$ \\
\hline
\end{tabular}
\caption{Hyperparameter search grids for Setting~III. Baseline-ERM and Baseline-DRO include architecture search; all other baselines use a fixed $[128, 128]$ architecture. $^*[256,256,128]$ included only for $d_A \geq 8$. $^\dagger$IW-KMM, IW-Classify, and all PL variants.}
\label{tab:tuning-grids-III}
\end{table}

For the DRUM variants, all hyperparameters are kept fixed (identical to Setting~II) across $d_A$, adjusting only the generator latent dimension as $q = \max(4, d_A)$. For the original baselines (ERM and DRO), a grid search over network architecture, learning rate, and regularization parameters is performed at each $d_A$ to ensure a fair comparison as problem complexity increases (Table~\ref{tab:tuning-grids-III}). For the additional baselines (IW-KMM, IW-Classify, and all pseudo-label variants), the same tuning grids as Settings~I and~II (Table~\ref{tab:tuning-grids-settingI}) were used with a fixed $[128, 128]$ architecture, and hyperparameters were re-tuned at each $d_A$. Across all $d_A$ values, the IW and pseudo-label baselines consistently selected lr $= 10^{-3}$ with epochs between 20 and 50; all DRO pseudo-label variants selected $\rho = 0.25$ (with a single exception of $\rho = 0.5$ for PL-MICE+DRO at $d_A = 9$), suggesting the negligible benefit of the DRO penalty observed in Settings~I and~II.

\subsection{Real-world Data Experiment details}
\label{app:realdat-details}

\paragraph{Data Details}
The source population comprised out-of-hospital cardiac arrest (OHCA) patients treated by emergency medical services (EMS) providers, as recorded in the Resuscitation Outcomes Consortium (ROC) Cardiac Epidemiologic Registry (Epistry) (Version 3, covering the period from April 1, 2011, to June 30, 2015). The ROC, a North American database established in 2004, aims to advance clinical research on cardiopulmonary arrest. Ethical approval was obtained from the National University of Singapore Institutional Review Board (IRB), which granted an exemption for this study (IRB Reference Number: NUS-IRB-2023-451).
The Pan-Asian Resuscitation Outcomes Study (PAROS) data were approved by the relevant ethics committees at each participating site and by the Centralized Institutional Review Board and Domain Specific Review Board for Singapore (reference numbers: 2010-270-C, C-10-545, 2013/604/C, and 2013/00929). Informed consent was waived for both datasets due to the retrospective, observational nature of the study, and all data were de-identified prior to analysis.
The target population comprised adult OHCA patients recorded in the Singapore PAROS registry from April 2010 to December 2016. The external validation cohorts comprised adult OHCA patients recorded in PAROS from January 2011 to December 2014 in South Korea and January 2009 to December 2014 in Japan.

For cohort formation, we retained adult patients (aged 18 and older) with known sex, determinable initial cardiac rhythm, and non-missing values for all shared covariates. Continuous variables were standardized using source population statistics. 
We note that EMS response time is recorded at some PAROS sites, though not uniformly across all deployment populations. Because epinephrine dose and blood pH are consistently absent across all PAROS sites, and our framework requires a common set of structurally missing covariates across deployment populations, we treat all three prehospital variables as structurally missing for consistency.

Table~\ref{tab:baseline} presents baseline characteristics across all study populations. Several notable differences are apparent. The source population (US-ROC) was younger on average (mean age $63.8$ years) than the Asian cohorts, with Japan the oldest (mean age $72.4$ years). The rate of shockable initial rhythm was highest in the US source ($43.4\%$) and considerably lower across the Asian cohorts ($7.0$--$17.4\%$). Favorable neurological outcome rates varied substantially, from $19.3\%$ in the US source to $3.1\%$ in the target Singapore population and $3.4$--$4.2\%$ across the external cohorts. These differences in outcome rates and covariate distributions illustrate the challenge of cross-regional clinical prediction and motivate transfer learning methods robust to population shift.

\textbf{Source data diagnostic: $A \mid X$ dependence and predictive value.}
The relationship between the structurally missing prehospital care variables $A$ and the stable covariates $X$ in the US source data is summarized as follows. 
The conditional dependence of $A$ on $X$ is weak: the mean $R^2$ from regressing each $A_j$ on all $X$ is 0.035, and random forest nonlinear prediction yields negative $R^2$, indicating that $A$ is weakly predictable from $X$ using the evaluated models. 
By contrast, $A$ contributes substantially to outcome prediction: including $A$ improves AUROC by approximately 0.07 across architectures, and SHAP analysis attributes 40\% of total feature importance to $A$. 
This combination of weak $A \mid X$ dependence and strong predictive value motivates the use of DRUM: the structurally missing variables carry information that is not readily recovered from $X$ alone, yet their distribution in the target population is unknown.

\begin{table}[!htbp]
\centering
\caption{Hyperparameter search grids for the OHCA application. All hyperparameters were selected by minimizing Brier score on held-out validation data.}
\label{tab:hp-realdata}
\footnotesize
\setlength{\tabcolsep}{4pt}
\begin{tabular}{llll}
\toprule
Component & Parameter & Search range & Selected \\
\midrule
\multirow{3}{*}{Outcome model $\hat{f}$}
& architecture & $[16],\; [32],\; [16,16],\; [32,32],\; [64,64],\; [128,128]$ & $[32]$ \\
& learning rate & $\{10^{-4}, 5 \times 10^{-4}, 10^{-3}, 5 \times 10^{-3}\}$ & $5 \times 10^{-3}$ \\
& epochs & $\{50, 100, 200\}$ & $50$ \\
\midrule
\multirow{7}{*}{DRUM generator}
& $\delta$ & $\{0.1, 0.3, 0.5\}$ & $0.3$ \\
& hidden dim & $\{4, 16\}$ & $16$ \\
& noise dim & $\{4, 8\}$ & $4$ \\
& $\text{lr}_{\text{primal}}$ & $\{5 \times 10^{-6}, 10^{-5}, 5 \times 10^{-5}\}$ & $10^{-5}$ \\
& $\text{lr}_{\text{dual}}$ & $\{5 \times 10^{-5}, 10^{-4}\}$ & $10^{-4}$ \\
& engression epochs & $\{50, 100, 200\}$ & $200$ \\
& finetune steps & $\{80, 100, 150\}$ & $150$ \\
\midrule
\multirow{4}{*}{DRUM debiasing}
& $\text{lr}_{\omega}$ & $\{10^{-4}, 5 \times 10^{-4}\}$ & $5 \times 10^{-4}$ \\
& $\omega$ epochs & $\{50, 100, 200, 300\}$ & $200$ \\
& $F_x$ epochs & $\{10, 20, 30, 40\}$ & $40$ \\
& $\text{lr}_{F_x}$ & $\{10^{-4}\}$ & $10^{-4}$ \\
\midrule
\multirow{3}{*}{Baseline-ERM}
& architecture & $[32],\; [64],\; [32, 32]$ & $[64]$ \\
& learning rate & $\{10^{-3}, 5 \times 10^{-3}\}$ & $5 \times 10^{-3}$ \\
& epochs & $\{30, 50, 100\}$ & $100$ \\
\midrule
\multirow{4}{*}{Baseline-DRO}
& architecture & $[32],\; [64],\; [32, 32]$ & $[64]$ \\
& learning rate & $\{10^{-3}, 5 \times 10^{-3}\}$ & $5 \times 10^{-3}$ \\
& epochs & $\{30, 50, 100\}$ & $100$ \\
& $\rho$ & $\{0.25, 0.5, 1.0\}$ & $0.25$ \\
\midrule
\multirow{3}{*}{IW-KMM}
& architecture & $[32],\; [64],\; [32, 32]$ & $[64]$ \\
& learning rate & $\{10^{-3}, 5 \times 10^{-3}\}$ & $5 \times 10^{-3}$ \\
& epochs & $\{30, 50, 100\}$ & $100$ \\
\midrule
\multirow{3}{*}{IW-Classify}
& architecture & $[32],\; [64],\; [32, 32]$ & $[64]$ \\
& learning rate & $\{10^{-3}, 5 \times 10^{-3}\}$ & $5 \times 10^{-3}$ \\
& epochs & $\{30, 50, 100\}$ & $50$ \\
\bottomrule
\end{tabular}
\end{table}

\paragraph{Hyperparameter fine-tuning.}
For DRUM, we use the three-stage estimation procedure described in Section~\ref{sec:estimation}. In Stage~1, the outcome model $\hat{f}$ is trained on source data. In Stage~2, the generator and worst-case search parameters are tuned. 
In the subsequent bias-correction step, the debiasing parameters for the density-ratio estimator and final estimator are tuned with the generator fixed from Stage~2.
The density ratio $\hat{\omega}$ was estimated via probabilistic classification within each cross-fitting fold (Algorithm~\ref{alg:proposed-Debiased}), with weights clipped to a finite upper bound and normalized to unit mean within each fold. 
DRUM is unsupervised in that the estimator, the outcome model, the worst-case generator, and the bias correction all use no target outcome labels and can be run with default hyperparameters. As with most methods, hyperparameter selection benefits from validation; here a small held-out labeled Singapore subset was used solely for final selection by Brier score, while no target labels entered the outcome model, generator, or bias correction. 
For the remaining Singapore observations, the covariates $X$ were used for unsupervised generator training while the outcome labels were held out for evaluation. The entire Japan and Korea cohorts were held out from both training and hyperparameter selection.

Table~\ref{tab:hp-realdata} summarizes the search grids for all methods. 
IW-KMM uses a Gaussian kernel with bandwidth selected by the median heuristic ($\sigma = 1.277$), $B = 1{,}000$. IW-Classify uses logistic regression for domain classification (AUC $= 0.649$).
Pseudo-label baselines (PL-Mean, PL-MICE, PL-MissForest, each with ERM and DRO) were tuned over the same architecture and learning rate grids as Baseline-ERM and Baseline-DRO, with DRO variants additionally searching over $\rho \in \{0.25, 0.5, 1.0\}$. All six variants selected architecture $[64]$ and learning rate $5 \times 10^{-3}$; epochs ranged from 50 to 100 across variants.

\paragraph{Calibration curves for all methods.}
Figure~\ref{fig:cal_all_appendix} presents calibration curves for all 10 baselines and DRUM. The corresponding main-text figure displays only DRUM and Baseline-ERM, the baseline with the lowest ECE at all three sites, because overlaying all 11 methods obscures the individual bin-level patterns.

\begin{figure}[!htbp]
    \centering
    \includegraphics[width=\linewidth]{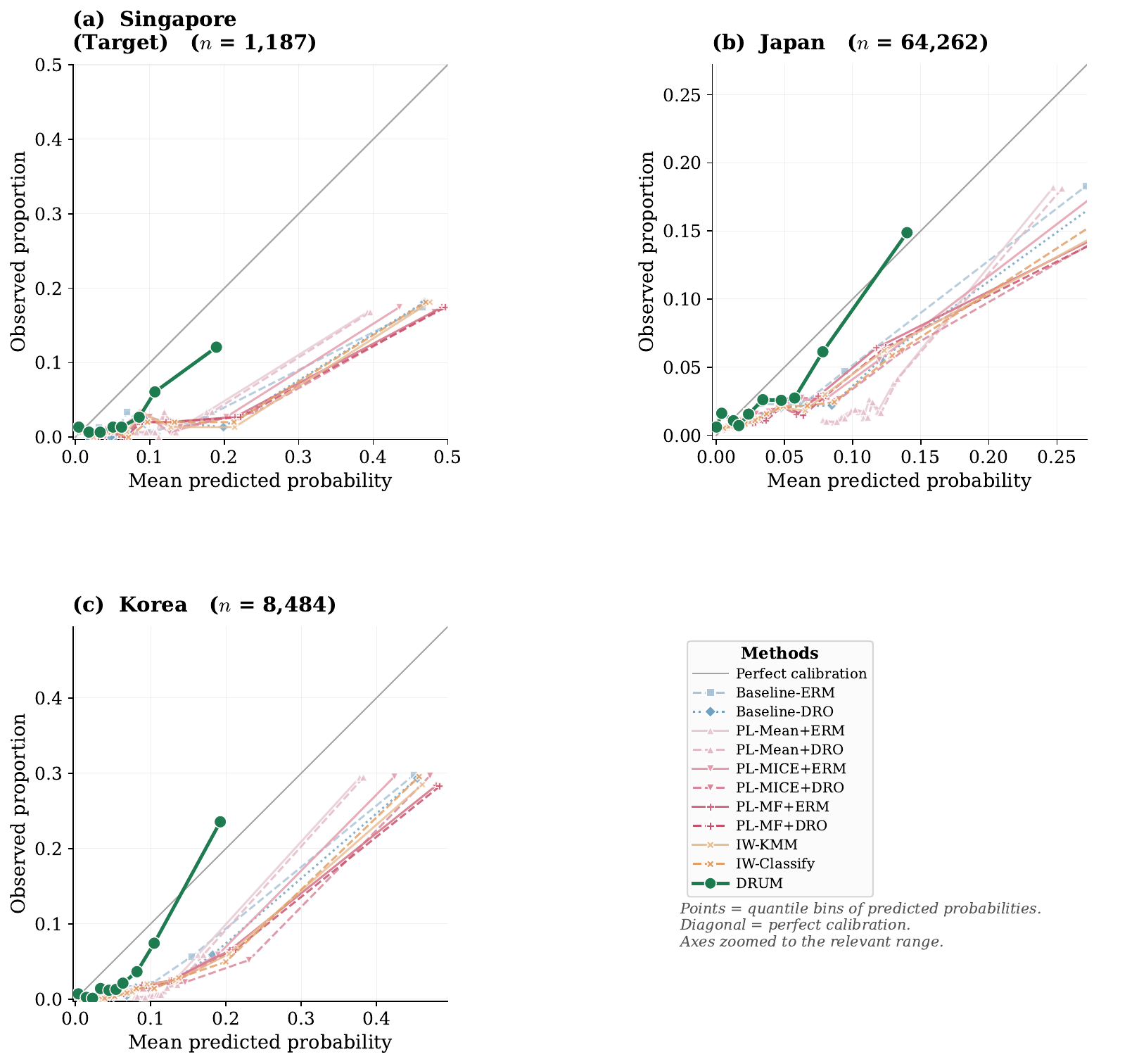}
    \caption{Calibration plots for all 10 baselines and DRUM across three OHCA populations. Points show mean predicted probability versus observed outcome proportion within quantile bins, using 8 bins for Singapore and 10 bins for Japan and Korea. The diagonal line represents perfect calibration, and axes are truncated to the relevant ranges for visibility.}
    \label{fig:cal_all_appendix}
\end{figure}

\paragraph{Discrimination Analysis}

Table~\ref{tab:discrimination_appendix} reports discrimination metrics across all sites. The baselines and DRUM (plug-in) achieve comparable AUROC, consistent with the relative ranking of patient risk by the stable covariates $X$ being largely preserved across populations.
The bias-corrected DRUM shows lower AUROC and AUPRC while achieving the lowest calibration error (Figure~\ref{fig:ohca_ece}), trading some rank-based discrimination for improved calibration.
Because clinical decisions that rely on absolute risk estimates, such as thresholding predicted probability to guide treatment, require calibrated rather than merely well-ranked predictions, this trade-off favors the property of primary clinical relevance.

Table~\ref{tab:fixed-cutoff} evaluates classification performance at fixed probability cutoffs $t \in \{0.03, 0.05\}$, chosen near the observed outcome prevalence in the Asian OHCA populations ($3$--$5\%$) to reflect clinically meaningful operating points. At these thresholds the bias-corrected DRUM is competitive with the best baseline on F1, precision, and specificity across all three sites, and attains the best or tied-best performance in Korea. The lower AUROC of the bias-corrected DRUM therefore does not translate into materially worse classification at the decision thresholds relevant to deployment: because its predicted probabilities are calibrated to the low target prevalence, a fixed clinically motivated cutoff yields classification comparable to the baselines despite the reduction in global rank separation.

\begin{table}[!htbp]
\centering
\caption{AUROC and AUPRC values across three OHCA populations. 95\% bootstrap CIs ($B = 2{,}000$). }
\label{tab:discrimination_appendix}
\footnotesize
\setlength{\tabcolsep}{2pt}
\begin{tabularx}{\textwidth}{ l *{3}{>{\centering\arraybackslash}X} }
\toprule
\multicolumn{4}{c}{\textbf{AUROC}} \\
\midrule
Method & Singapore (Target) & Japan & Korea \\
\midrule
Baseline-ERM         & 0.838 (0.765, 0.904) & 0.810 (0.800, 0.820) & 0.909 (0.893, 0.924) \\
Baseline-DRO         & 0.839 (0.764, 0.904) & 0.809 (0.799, 0.819) & 0.910 (0.895, 0.924) \\
IW-KMM               & 0.838 (0.762, 0.908) & 0.804 (0.794, 0.814) & 0.904 (0.889, 0.919) \\
IW-Classify          & 0.843 (0.768, 0.907) & 0.806 (0.797, 0.816) & 0.906 (0.891, 0.922) \\
PL-Mean+ERM          & 0.829 (0.749, 0.902) & 0.775 (0.763, 0.786) & 0.900 (0.882, 0.918) \\
PL-Mean+DRO          & 0.831 (0.751, 0.902) & 0.778 (0.767, 0.789) & 0.900 (0.882, 0.918) \\
PL-MICE+ERM          & 0.845 (0.777, 0.906) & 0.820 (0.811, 0.829) & 0.910 (0.895, 0.925) \\
PL-MICE+DRO          & 0.846 (0.778, 0.908) & 0.816 (0.807, 0.825) & 0.908 (0.892, 0.922) \\
PL-MF+ERM            & 0.844 (0.772, 0.906) & 0.805 (0.795, 0.815) & 0.903 (0.888, 0.918) \\
PL-MF+DRO            & 0.844 (0.773, 0.907) & 0.805 (0.795, 0.815) & 0.903 (0.888, 0.918) \\
\addlinespace
DRUM (plug-in)       & 0.838 (0.769, 0.903) & 0.816 (0.806, 0.826) & 0.911 (0.896, 0.926) \\
DRUM                 & 0.777 (0.695, 0.854) & 0.769 (0.758, 0.780) & 0.853 (0.833, 0.874) \\
\midrule
\multicolumn{4}{c}{\textbf{AUPRC}} \\
\midrule
Method & Singapore (Target) & Japan & Korea \\
\midrule
Baseline-ERM         & 0.195 (0.124, 0.317) & 0.198 (0.183, 0.214) & 0.349 (0.303, 0.402) \\
Baseline-DRO         & 0.185 (0.120, 0.295) & 0.194 (0.180, 0.210) & 0.341 (0.296, 0.393) \\
IW-KMM               & 0.206 (0.131, 0.336) & 0.172 (0.159, 0.187) & 0.318 (0.276, 0.371) \\
IW-Classify          & 0.190 (0.123, 0.303) & 0.190 (0.175, 0.206) & 0.337 (0.293, 0.391) \\
PL-Mean+ERM          & 0.212 (0.131, 0.354) & 0.191 (0.177, 0.208) & 0.342 (0.299, 0.400) \\
PL-Mean+DRO          & 0.215 (0.131, 0.360) & 0.192 (0.177, 0.208) & 0.340 (0.298, 0.398) \\
PL-MICE+ERM          & 0.192 (0.124, 0.312) & 0.196 (0.182, 0.212) & 0.347 (0.302, 0.401) \\
PL-MICE+DRO          & 0.194 (0.124, 0.312) & 0.192 (0.177, 0.208) & 0.343 (0.299, 0.397) \\
PL-MF+ERM            & 0.171 (0.115, 0.275) & 0.172 (0.159, 0.186) & 0.318 (0.276, 0.369) \\
PL-MF+DRO            & 0.179 (0.117, 0.287) & 0.172 (0.159, 0.187) & 0.320 (0.277, 0.371) \\
\addlinespace
DRUM (plug-in)       & 0.187 (0.121, 0.311) & 0.199 (0.185, 0.215) & 0.342 (0.299, 0.400) \\
DRUM                 & 0.127 (0.077, 0.225) & 0.146 (0.134, 0.159) & 0.253 (0.214, 0.302) \\
\bottomrule
\multicolumn{4}{@{}p{\textwidth}@{}}{\footnotesize DRUM (plug-in): estimator without bias correction (Algorithm~\ref{alg:local}). DRUM: estimator with bias correction (Algorithm~\ref{alg:proposed-Debiased}).} \\
\end{tabularx}
\end{table}

\begin{table}[!htbp]
\centering
\caption{Classification performance at fixed probability cutoffs across three
OHCA populations. Cutoff $t$: classify $\hat{Y}=1$ if the predicted probability
is at least $t$. 
DRUM (plug-in): estimator without bias correction (Algorithm~\ref{alg:local}). 
DRUM: estimator with bias correction (Algorithm~\ref{alg:proposed-Debiased})}
\label{tab:fixed-cutoff}
\footnotesize
\setlength{\tabcolsep}{2pt}
\renewcommand{\arraystretch}{1.05}

\begin{tabularx}{\textwidth}{
    c
    l
    *{3}{>{\centering\arraybackslash}X}
    *{3}{>{\centering\arraybackslash}X}
    *{3}{>{\centering\arraybackslash}X}
}
\toprule
& &
\multicolumn{3}{c}{Singapore (Target)} &
\multicolumn{3}{c}{Japan} &
\multicolumn{3}{c}{Korea} \\
\cmidrule(lr){3-5}
\cmidrule(lr){6-8}
\cmidrule(lr){9-11}
$t$ & Method
& F1 & Prec & Spec
& F1 & Prec & Spec
& F1 & Prec & Spec \\
\midrule

\multirow{12}{*}{0.03}
& Baseline-ERM
& \textbf{0.084} & \textbf{0.044} & \textbf{0.303}
& \textbf{0.116} & \textbf{0.062} & \textbf{0.543}
& \textbf{0.112} & \textbf{0.059} & 0.316 \\

& Baseline-DRO
& 0.075 & 0.039 & 0.190
& 0.097 & 0.052 & 0.404
& 0.097 & 0.051 & 0.191 \\

& IW-KMM
& 0.080 & 0.042 & 0.239
& 0.104 & 0.055 & 0.461
& 0.102 & 0.054 & 0.241 \\

& IW-Classify
& 0.075 & 0.039 & 0.180
& 0.095 & 0.050 & 0.378
& 0.095 & 0.050 & 0.171 \\

& PL-Mean+ERM
& 0.064 & 0.033 & 0.000
& 0.067 & 0.034 & 0.000
& 0.080 & 0.042 & 0.000 \\

& PL-Mean+DRO
& 0.064 & 0.033 & 0.000
& 0.067 & 0.034 & 0.000
& 0.080 & 0.042 & 0.000 \\

& PL-MICE+ERM
& 0.082 & 0.043 & 0.258
& 0.111 & 0.059 & 0.498
& 0.107 & 0.056 & 0.279 \\

& PL-MICE+DRO
& 0.080 & 0.042 & 0.243
& 0.108 & 0.057 & 0.472
& 0.103 & 0.055 & 0.251 \\

& PL-MF+ERM
& 0.076 & 0.040 & 0.197
& 0.098 & 0.052 & 0.408
& 0.097 & 0.051 & 0.194 \\

& PL-MF+DRO
& 0.075 & 0.039 & 0.180
& 0.094 & 0.050 & 0.379
& 0.095 & 0.050 & 0.174 \\

\addlinespace

& DRUM (plug-in)
& 0.076 & 0.039 & 0.193
& 0.099 & 0.052 & 0.413
& 0.099 & 0.052 & 0.209 \\

& DRUM
& 0.082 & 0.043 & 0.296
& 0.110 & 0.059 & 0.525
& \textbf{0.112} & \textbf{0.059} & \textbf{0.327} \\

\midrule

\multirow{12}{*}{0.05}
& Baseline-ERM
& \textbf{0.100} & \textbf{0.053} & \textbf{0.440}
& \textbf{0.142} & \textbf{0.078} & \textbf{0.682}
& 0.138 & \textbf{0.075} & 0.478 \\

& Baseline-DRO
& 0.088 & 0.046 & 0.332
& 0.122 & 0.066 & 0.580
& 0.118 & 0.063 & 0.358 \\

& IW-KMM
& 0.090 & 0.047 & 0.348
& 0.123 & 0.066 & 0.586
& 0.120 & 0.064 & 0.373 \\

& IW-Classify
& 0.088 & 0.046 & 0.337
& 0.120 & 0.065 & 0.573
& 0.117 & 0.062 & 0.354 \\

& PL-Mean+ERM
& 0.064 & 0.033 & 0.000
& 0.067 & 0.035 & 0.000
& 0.080 & 0.042 & 0.000 \\

& PL-Mean+DRO
& 0.064 & 0.033 & 0.000
& 0.067 & 0.034 & 0.000
& 0.080 & 0.042 & 0.000 \\

& PL-MICE+ERM
& 0.098 & 0.051 & 0.389
& 0.135 & 0.074 & 0.633
& 0.128 & 0.069 & 0.418 \\

& PL-MICE+DRO
& 0.093 & 0.049 & 0.356
& 0.131 & 0.071 & 0.613
& 0.124 & 0.066 & 0.394 \\

& PL-MF+ERM
& 0.090 & 0.047 & 0.353
& 0.124 & 0.067 & 0.593
& 0.121 & 0.064 & 0.376 \\

& PL-MF+DRO
& 0.087 & 0.046 & 0.327
& 0.120 & 0.065 & 0.565
& 0.116 & 0.062 & 0.346 \\

\addlinespace

& DRUM (plug-in)
& 0.090 & 0.047 & 0.328
& 0.122 & 0.066 & 0.573
& 0.118 & 0.063 & 0.353 \\

& DRUM
& 0.091 & 0.048 & 0.429
& 0.133 & 0.073 & 0.674
& \textbf{0.139} & \textbf{0.075} & \textbf{0.509} \\

\bottomrule

\end{tabularx}
\end{table}

\paragraph{Diagnostic: worst-case $A \mid X$ recovered by DRUM generators.}
Since $A$ is structurally absent in the target Asian cohorts, direct validation against the target $A \mid X$ distribution is impossible. We instead examine the behavior of the DRUM generators relative to the source $A \mid X$ distribution.
Figure~\ref{fig:AX_marginal} compares the marginal distributions of each $A$ variable under three scenarios: the source-observed distribution, samples from the DRUM worst-case generator (pooled over $X$), and samples from the source-fit generator. Both generators remain close to the source marginal.

\begin{figure}[!htbp]
    \centering
    \includegraphics[width=\linewidth]{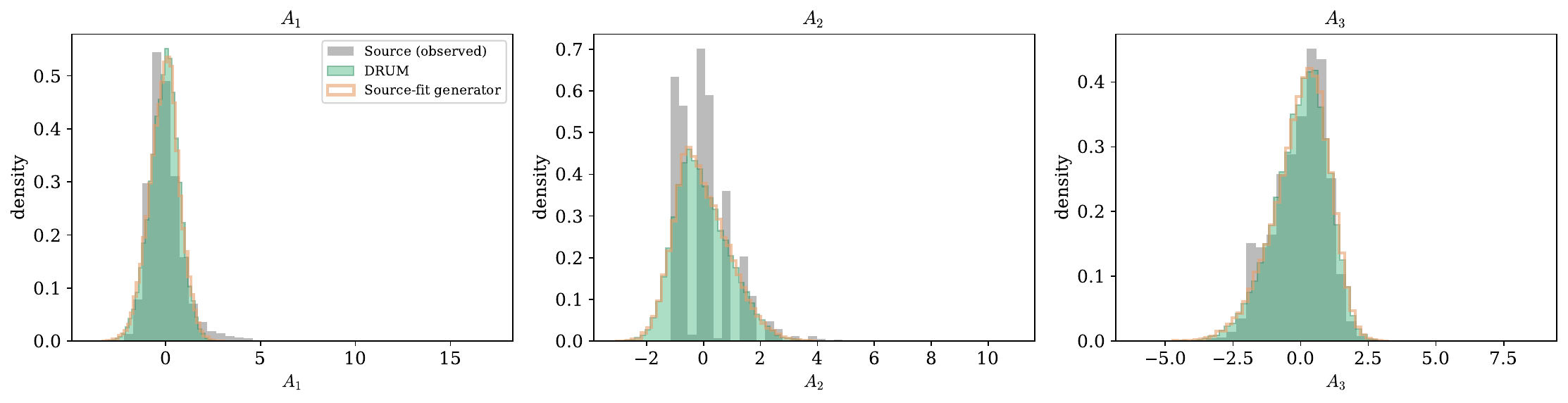}
    \caption{Marginal distributions of each structurally missing covariate $A$ in the source (observed) compared with samples from the DRUM worst-case generator (pooled over $X$) and the source-fit generator. Source $A$ has approximately zero mean and unit standard deviation due to standardization.}
    \label{fig:AX_marginal}
\end{figure}

\paragraph{Local Supervised Baselines (Oracle Reference)}
As an additional reference, we trained local ERM models using data from sites with sufficient sample size. For each site (Singapore, Japan, and Korea), we performed a stratified 80/20 train-test split and trained three standard models on the local training set:
logistic regression (with 5-fold cross-validation over the regularization parameter $C \in \{0.01, 0.1, 1, 10\}$), random forest (with 5-fold cross-validation over tree depth, number of estimators, and minimum leaf size), and a two-layer neural network (with grid search over 50 architecture-learning rate-epoch combinations, selected by validation loss). All models were evaluated on the held-out 20\% local test set with 2{,}000 bootstrap resamples for confidence intervals. Table~\ref{tab:local-baselines} reports the results: the local models achieved similar performance to each other, with neural networks and random forests performing comparably to logistic regression. 
Note that this comparison is asymmetric: DRUM is designed for unsupervised transfer learning setting and therefore the training uses no local labeled data, while the local models are trained on 80\% of local labeled data and evaluated on the remaining 20\%. The comparison is intended to contextualize DRUM's transfer performance relative to models with direct access to local labels.

\begin{table}[!htbp]
\centering
\caption{Performance of local baseline models trained and evaluated on site-specific data (80/20 stratified split). Metrics are reported as point estimate (95\% bootstrap CI). These models use local labeled data for training.}
\label{tab:local-baselines}
\small
\begin{tabular}{llccc}
\toprule
Site & Method & Brier Score & AUROC & AUPRC \\
\midrule
\multirow{3}{*}{Singapore}
& Logistic Regression & 0.0287 (0.0213, 0.0371) & 0.8264 (0.7563, 0.8882) & 0.1628 (0.0941, 0.2846) \\
& Random Forest       & 0.0283 (0.0210, 0.0370) & 0.8183 (0.7407, 0.8836) & 0.1852 (0.1021, 0.3213) \\
& Neural Network      & 0.0284 (0.0211, 0.0369) & 0.8299 (0.7602, 0.8906) & 0.1741 (0.1016, 0.2985) \\
\midrule
\multirow{3}{*}{Japan}
& Logistic Regression & 0.0298 (0.0273, 0.0323) & 0.8607 (0.8443, 0.8772) & 0.2080 (0.1779, 0.2482) \\
& Random Forest       & 0.0296 (0.0271, 0.0322) & 0.8597 (0.8424, 0.8763) & 0.2158 (0.1832, 0.2548) \\
& Neural Network      & 0.0295 (0.0270, 0.0320) & 0.8639 (0.8477, 0.8800) & 0.2151 (0.1841, 0.2558) \\
\midrule
\multirow{3}{*}{Korea}
& Logistic Regression & 0.0338 (0.0273, 0.0410) & 0.8798 (0.8397, 0.9167) & 0.3076 (0.2188, 0.4176) \\
& Random Forest       & 0.0335 (0.0269, 0.0407) & 0.8712 (0.8243, 0.9142) & 0.3079 (0.2188, 0.4169) \\
& Neural Network      & 0.0338 (0.0271, 0.0408) & 0.8788 (0.8333, 0.9191) & 0.3147 (0.2235, 0.4236) \\
\bottomrule
\end{tabular}
\end{table}
 
\end{document}